\newif\ifarxiv
\arxivtrue % FLIP for ITW submission

\ifarxiv
  \documentclass[draftnocls,journal,onecolumn]{IEEEtran}
\else
  \documentclass[conference]{IEEEtran}
\usepackage{pgfplots}
\pgfplotsset{compat=1.18}
  \IEEEoverridecommandlockouts
\fi

\usepackage{amsthm,amsmath,amssymb,latexsym,etoolbox,graphicx,tikz,pgfplots,multicol,xcolor}
\usepackage[shortlabels]{enumitem}
\setlist[itemize]{leftmargin=*,labelsep=0.4em,topsep=1pt,partopsep=0pt,parsep=0pt,itemsep=0pt}
\setlist[enumerate]{leftmargin=*,labelsep=0.4em,align=right,topsep=1pt,partopsep=0pt,parsep=0pt,itemsep=0pt}
\usepackage{caption,subcaption,mathtools,mathrsfs,bbm}
\usepackage[T1]{fontenc}
\usepackage[noadjust]{cite}
\usepackage[ruled,lined,longend]{algorithm2e}
\usepackage{booktabs,array,colortbl}
\newcommand{\faintrule}{\arrayrulecolor{black!25}\midrule\arrayrulecolor{black}}
\usepackage{url}
\ifarxiv
  \usepackage[hidelinks]{hyperref}
\fi
\mathtoolsset{showonlyrefs}
\usetikzlibrary{positioning,chains,fit,shapes,calc,patterns,decorations.pathreplacing}
\allowdisplaybreaks

\theoremstyle{plain}
\newtheorem{theorem}{Theorem}
\newtheorem{lemma}[theorem]{Lemma}
\newtheorem{proposition}[theorem]{Proposition}

\newtheorem{claim}[theorem]{Claim}
\theoremstyle{definition}
\newtheorem{definition}[theorem]{Definition}

\theoremstyle{remark}
\newtheorem{remark}[theorem]{Remark}

\newcommand{\cA}{\mathscr{A}}
\newcommand{\cB}{\mathscr{B}}

\newcommand{\cM}{\mathscr{M}}
\newcommand{\cP}{\mathscr{P}}
\newcommand{\cS}{\mathscr{S}}
\newcommand{\cX}{\mathscr{X}}
\newcommand{\cY}{\mathscr{Y}}

\newcommand{\cD}{\mathscr{D}}
\newcommand{\cF}{\mathscr{F}}
\newcommand{\cG}{\mathscr{G}}

\newcommand{\Prob}{\mathsf{P}}
\newcommand{\Probc}{\mathsf{Q}}
\newcommand{\Probg}{\mathsf{G}}
\newcommand{\Proba}{\tilde{\mathsf{P}}}
\newcommand{\Exp}{\mathbb{E}}
\newcommand{\fP}{\sigma}
\newcommand{\szero}{s_0}

\newcommand{\chn}{W_{Y|X,S}}
\newcommand{\set}[1]{\left\{#1\right\}}
\newcommand{\accept}{\mathsf{ACCEPT}}
\newcommand{\reject}{\mathsf{REJECT}}

\newcommand{\PNA}{\mathrm{na}}
\newcommand{\PADV}{\mathrm{adv}}
\newcommand{\condfull}{M,\hat M,\tilde M,Y^{\tau-1}\!,S^{\tau-1}}
\newcommand{\condnoadv}{M,Y^{\tau-1}\!,S^{\tau-1}}
\newcommand{\condmsg}{M,X^{i-1}\!,Y^{i-1}}
\newcommand{\condmsgseq}{M,X^{\tau-1}\!,Y^{\tau-1}}

\newcommand{\posnoadv}{\Prob^{\PNA}_{M\mid Y^t\!,S^t}}
\newcommand{\encno}{\Prob_{X_\tau\mid \condnoadv}}

\newcommand{\Tg}{T_\gamma}
\newcommand{\PgM}{\Probg_{M\mid Y^t\!,S^t}}

\newcommand{\authdec}{D}

\newcommand{\parnoindent}[1]{\noindent {\it #1:}}
\ifarxiv\fi

\ifarxiv
  \newcommand{\extref}[1]{Appendix~\ref{#1}}
  
  \newcommand{\extapps}{the appendices}
  \newcommand{\eqbr}{}
\else
  \newcommand{\extref}[1]{the extended version~\cite{thispaper-ext}}

  \newcommand{\extapps}{the extended version~\cite{thispaper-ext}}
  
  \newcommand{\eqbr}{\\&\quad}
\fi

\begin{document}
\title{Authentication over Arbitrarily Varying Channels with Causal Adversaries}

\ifarxiv
\author{%
\normalfont
\renewcommand{\arraystretch}{0.75}%
\begin{tabular}{@{}c@{\hspace{2em}}c@{}}
{ Mayank Bakshi} & { Vinod M.\ Prabhakaran} \\
\small\itshape School of Informatics, Computing, and Cyber Systems & \small\itshape School of Technology and Computer Science \\
\small\itshape Northern Arizona University & \small\itshape Tata Institute of Fundamental Research \\
\small Flagstaff, AZ, USA & \small Mumbai, India \\
\small\texttt{mayank.bakshi@nau.edu} & \small\texttt{vinodmp@tifr.res.in} \\[3.2ex]
{ Bikash Kumar Dey} & { Oliver Kosut} \\
\small\itshape Department of Electrical Engineering & \small\itshape School of Electrical, Computer, and Energy Engineering \\
\small\itshape Indian Institute of Technology Bombay & \small\itshape Arizona State University \\
\small Mumbai, India & \small Tempe, AZ, USA \\
\small\texttt{bikash@ee.iitb.ac.in} & \small\texttt{okosut@asu.edu}
\end{tabular}
\thanks{The work of M.\ Bakshi is based upon work supported by the National Science Foundation under Grant No.\ CCF-2107526. This work was also funded under the State of Arizona Technology and Research Initiative Fund (TRIF), administered by the Arizona Board of Regents. V.\ M.\ Prabhakaran acknowledges support of the DAE, Govt.\ of India, under project no.\ RTI4014 and of ANRF, Govt.\ of India, through project ANRF/ARGM/2025/000821/TS. The work of B.\ K.\ Dey was partly supported by the Bharti Centre for Communication at IIT Bombay. The work of O.\ Kosut is supported in part by NSF grant CIF-2312666.}%
}
\else
\author{%
\IEEEauthorblockN{%
Mayank Bakshi\IEEEauthorrefmark{1},\
Vinod M.\ Prabhakaran\IEEEauthorrefmark{2},\
Bikash Kumar Dey\IEEEauthorrefmark{3},\
Oliver Kosut\IEEEauthorrefmark{4}%
}%
\thanks{%
\IEEEauthorrefmark{1}\textit{School of Informatics, Computing, and Cyber Systems, Northern Arizona University},\ \texttt{mayank.bakshi@nau.edu};\
\IEEEauthorrefmark{2}\textit{School of Technology and Computer Science, Tata Institute of Fundamental Research},\ \texttt{vinodmp@tifr.res.in};\
\IEEEauthorrefmark{3}\textit{Department of Electrical Engineering, IIT Bombay},\ \texttt{bikash@ee.iitb.ac.in};\
\IEEEauthorrefmark{4}\textit{School of Electrical, Computer, and Energy Engineering, Arizona State University},\ \texttt{okosut@asu.edu}.%
}%
}
\fi

\maketitle

\begin{abstract}
We study authentication over a discrete memoryless arbitrarily varying channel (AVC) in which the adversary selects the channel state causally based on past channel outputs. We prove that the authentication capacity is positive exactly when the channel is not \emph{distribution-overwritable} under stochastic encoding, and not \emph{I-overwritable} under deterministic encoding. In the stochastic case, whenever the authentication capacity is positive, it coincides with the no-adversary Shannon capacity. Both converses follow from a novel \emph{wait-and-overwrite} attack: the adversary tracks the no-adversary posterior over the message via the posterior guessing probability, and once it concentrates past a threshold, it samples a guess and a decoy from the posterior and overwrites the channel to mimic a no-adversary transmission of the decoy. The positivity results are obtained by constructing positive-rate codes with controlled overlap between codewords, together with a martingale concentration argument that handles adaptive adversarial strategies. For stochastic encoding, the full-capacity achievability further combines a capacity-achieving channel code with an authentication tag. These characterizations separate the causal stochastic-code and causal deterministic-code settings from each other and from the oblivious-adversary setting studied by Kosut and Kliewer.
\end{abstract}

\section{Introduction}
\label{sec:intro}

\noindent
In an authentication problem, the receiver either decodes the
transmitted message with confidence that it was not tampered with, or
raises an alarm. We consider authentication over an arbitrarily
varying channel (AVC) $\chn$ in which the state $S$ is chosen by an
adversary observing the channel output causally
(Fig.~\ref{fig:model}).   

For reliable communication over an AVC, it is shown
in~\cite{CsiszarN88} that the deterministic capacity is positive iff
$\chn$ is not \emph{symmetrizable}. An analogous characterization for
authentication is due to Kosut and Kliewer~\cite{KK18}. In order to model the authentication problem, ~\cite{KK18} assumes that the set of channel states $\cS$ (from which the adversary draws $S$) contains a nominal channel state $\szero$ that corresponds to the input-output transition probability when the adversarial interference is absent. They show that
authentication capacity against an oblivious adversary is positive
iff $\chn$ is not \emph{overwritable}, where $\chn$ is overwritable
if there is a conditional distribution $\Probc_{S\mid X'}$ such that,
for every $x,x'\!\in\!\cX$,
\begin{equation}
\sum_{s\in\cS}\Probc_{S\mid X'}(s\!\mid\!x')\,W(y\!\mid\!x,s)
   =W(y\!\mid\!x',\szero)\quad\forall y\!\in\!\cY.
\label{eq:obow_def}
\end{equation}
\label{def:obow}%
We refer to this as \emph{obliviously overwritable}, to distinguish it
from the related conditions below.
Sufficient conditions for positive capacity were later established for
myopic adversaries with non-causal noisy observations: $\chn$ not
I-overwritable~\cite{BeemerGKKY20mildly}, and $\chn$ not
distribution-overwritable~\cite{BakshiK:23ISIT}.

\paragraph*{This work}
We consider an \emph{output-feedback} adversary that, at each time
$t$, chooses the state $S_t$ as a function of $(S^{t-1},Y^{t-1})$.
This is the causal counterpart of the omniscient adversary and refines
the oblivious adversary of~\cite{KK18}. The corresponding
causal-adversary AVC reliability problem is treated
in~\cite{ChenJL19,DJLS13causal,ZJLS22causal,Langberg08}.
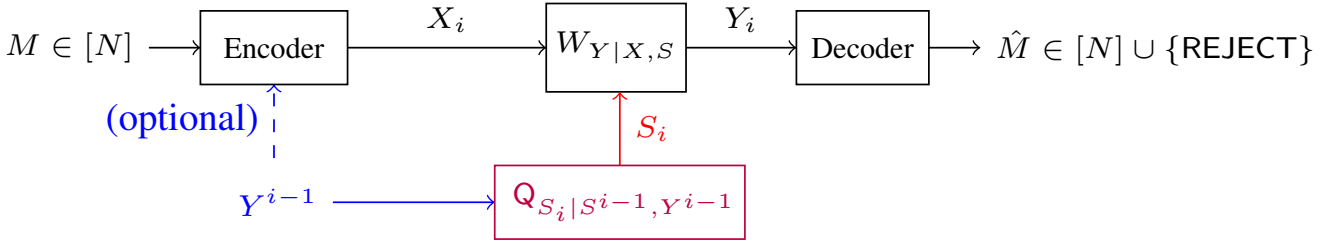
\begin{figure}[t]
\centering
\resizebox{0.95\columnwidth}{!}{%
\begin{tikzpicture}[scale=0.6]
\useasboundingbox (-4.2,4.1) rectangle (13.5,7.5);
\draw (-1.8,6.5) rectangle ++(2,1) node[pos=.5]{\scriptsize Encoder};
\draw[->,color=blue,dashed] (-.8, 5.5) -- node[left,font=\small]{\color{blue}(optional)} ++ (0, 1.0);
\draw[->,color=blue] (0, 4.9) node[left]{\scriptsize \color{blue} $Y^{i-1}$} -- ++(2.2,0);
\draw (2.9,6.4) rectangle ++(1.9,1.2) node[pos=.5]{\scriptsize $\chn$};
\draw[->,color=red] (3.9, 5.4) -- node[right]{\scriptsize \color{red} $S_i$} ++ (0, 1.0);
\draw[color=purple] (2.2, 4.4) rectangle ++(3.4,1) node[pos=0.5]{\scriptsize\color{purple} $\;\Probc_{S_i\mid S^{i-1},Y^{i-1}}$};
\draw[->] (-2.5,7) node[anchor=east]{\scriptsize $M\in[N]$} -- ++ (0.7,0);
\draw[->] (0.2, 7) -- node[above]{\scriptsize $X_i$} ++(2.7,0);
\draw[->] (4.8,7) -- node[above] {\scriptsize $Y_i$} ++ (1.5,0);
\draw (6.3, 6.5) rectangle ++(1.8,1) node[pos=0.5]{\scriptsize Decoder};
\draw[->] (8.1, 7) -- ++ (0.7,0)  node[right] {\scriptsize $\hat{M}\in[N]\cup\set{\reject}$};
\end{tikzpicture}}
\caption{Channel model. The adversary observes the channel output
$Y^{i-1}$ causally and selects state $S_i$. The dashed feedback link
to the encoder is optional.}
\label{fig:model}
\end{figure}
\begin{table}[b]
\centering
\caption{Authentication-capacity dichotomies.}
\label{tab:compare}
\renewcommand{\arraystretch}{0.95}
\setlength{\tabcolsep}{3pt}
\resizebox{.45\textwidth}{!}{%
\begin{tabular}{@{}llll@{}}
\toprule
Adversary's view & $C_{\mathrm{auth,stoch}}\!=\!0$ iff & $C_{\mathrm{auth,det}}\!=\!0$ iff & Ref. \\
\midrule
None (oblivious)         & oblivious-OW & oblivious-OW & \cite{KK18,KosutKITW18}\\
\faintrule
\begin{tabular}[t]{@{}l@{}} $Z^n\sim U_{Z|X}$ \\ (non-causal myopic) \end{tabular}       & \begin{tabular}[t]{@{}l@{}}myopic distribution-OW\\(multi-letter)\end{tabular} & open & \begin{tabular}[t]{@{}l@{}}\cite{BeemerKKGY19,BeemerGKKY20mildly},\\\cite{BakshiK:23ISIT,BakshiBeemerK:prep}\end{tabular}\\
\faintrule
$X^n$ (omniscient)       & I-OW & I-OW & \cite{BeemerGKKY20mildly}\\
\faintrule
\emph{$(Y^{t-1})$ (causal f.b.)} & \emph{distribution-OW} & \emph{I-OW} & \textbf{\emph{Thms.~\ref{thm:posstoch},~\ref{thm:posdet}}}\\
\bottomrule
\end{tabular}}
\end{table}
\paragraph*{Main results and organization}

We consider the authentication capacity with both stochastic and deterministic codes against
a causal adversary with output-feedback. Our
results are stated in terms of two channel conditions strictly weaker
than~\eqref{eq:obow_def}, called \emph{distribution-} and
\emph{I-overwritability} (formalized in Section~\ref{sec:main}). We
show that:
\begin{enumerate}[itemsep=0pt, label=(\emph{\roman*}), wide, labelindent=0pt]
\item the stochastic authentication capacity is positive iff
$\chn$ is not distribution-overwritable (Theorem~\ref{thm:posstoch});
\item the deterministic authentication capacity is positive iff
$\chn$ is not I-overwritable (Theorem~\ref{thm:posdet});
\item whenever the stochastic authentication capacity is
positive, it equals the no-adversary Shannon capacity $C_{\szero}$ (Theorem~\ref{thm:caprate}).
\end{enumerate}
Since the set of I-overwritable channels is strictly contained within the set of distribution overwritable channels (see Remark~\ref{rem:nest}), it follows that there are channels where the stochastic codes are necessary for achieving positive rates . The deterministic-code capacity, when positive, remains open, as it
does in the omniscient setting~\cite{BeemerGKKY20mildly}.
Table~\ref{tab:compare} places these results alongside the
oblivious, myopic, and omniscient dichotomies. A key contribution of our work is a wait-and-overwrite attack based on the posterior guessing probability of the
message; this enables the converse for both the stochastic and deterministic codes. Our achievability borrows ideas from~\cite{BakshiK:23ISIT}, modified for feedback adversaries through a martingale concentration argument.

Sections~\ref{sec:prelim}--\ref{sec:main} state the model and formal
results, Section~\ref{sec:converse} gives the converse, and
Sections~\ref{sec:positivity}--\ref{sec:capacity} sketch the
achievability. Detailed proofs and examples are in~\extapps.

\paragraph*{Other related work}

Authentication has been studied beyond the present adversary model.
It originates in
cryptography~\cite{SimmonsCRYPTO84,MaurerIT00,Maurer:1997} and has
been considered on noisy channels with shared keys, in keyless
schemes, and in Byzantine multiple-access
settings~\cite{LaiEPIT09,TuLIT18,GK16,PGYB18,GravesYS16,SangwanBDP19ISIT,S19M}.
Within the AVC framework, variants of the Kosut--Kliewer
characterization include AWGN authentication~\cite{GravesBKKY23},
structured codes~\cite{BKKGY19structured}, and additional myopic and
multiple-access models~\cite{BeemerKKGY19,BeemerGKKY20MAC}. Two-phase
constructions in this line are related to identification
codes~\cite{BassalygoBPPI96,AhlswedeD89,BurnashevIT00,BocheD19}, which
also underpin our achievability through the doubly-exponential message
growth of the authentication tag. Our wait-and-overwrite converse is also reminiscent of the
``babble-then-push'' attacks used to establish converses for AVCs
with \emph{online} causal adversaries observing past
inputs~\cite{ChenJL19,DJLS13causal,ZJLS22causal,Langberg08}. Some of the analytical tools we use also share a flavor with the
literature on adversarial hypothesis
testing~\cite{brandao2020adversarial,modak2023hypothesis,modak2024sequential,burnashev1976data},
though our setting and goal are different.

\section{Preliminaries}
\label{sec:prelim}

\subsection{Notation}
For a positive integer $N$, $[N]\triangleq\set{1,\ldots,N}$. Script
letters $\cX,\cY,\cS,\cM,\ldots$ denote finite alphabets, capital
letters the corresponding random variables, lowercase letters their
realizations, and $y^t\!=\!(y_1,\ldots,y_t)$. The set of probability
distributions on $\cA$ is $\mathcal{P}(\cA)$, $\mathbf{1}_a\!\in\!\mathcal{P}(\cA)$
is the point mass on $a$, and $\mathbb{V}(\cdot,\cdot)$ is the
variational distance.

\subsection{Channel and adversary}
$\chn$ is a discrete memoryless
channel with finite alphabets $\cX,\cS,\cY$ and a fixed
\emph{no-adversary state} $\szero\!\in\!\cS$;
$W_{Y\mid X,S=\szero}$ is the legitimate channel and $C_{\szero}$ its
Shannon capacity. At time $t$ the encoder emits
$X_t\!\in\!\cX$, the adversary picks $S_t\!\in\!\cS$, and
$Y_t\!\sim\!W(\cdot\!\mid\!X_t,S_t)$. The receiver outputs
$\hat M\!\in\![N]\!\cup\!\set{\reject}$, where $\reject$ flags the
adversary's presence. Throughout, $M\!\in\![N]$ is uniformly
distributed.

\subsection{Codes}
An $(N,n)$ \emph{code} is a sequential encoder $F=\set{\Prob_{X_t\mid M,X^{t-1},Y^{t-1}}}_{t=1}^n$ paired with a decoder $\phi\colon\cY^n\!\to\![N]\!\cup\!\set{\reject}$.
The code is \emph{deterministic} if each conditional law is a Dirac mass,
and \emph{stochastic} otherwise. The encoder is in the \emph{feedback}
(resp.\ \emph{no-feedback}) class according to whether $X_t$ may depend
on $Y^{t-1}$.

\subsection{Adversarial strategies and error metric}
An \emph{adversarial strategy} is a sequence
$\Probc=\set{\Probc_{S_t\mid S^{t-1},Y^{t-1}}}_t$; we write $\Probc_0$
for the trivial strategy $S_t\!\equiv\!\szero$. The
\emph{authentication error} of a code $(F,\phi)$ is
\begin{equation}
\varepsilon(F,\phi)\!\!\triangleq\!\! \Pr\nolimits_{F,\phi,\Probc_0}\!(\hat M\neq M)+ \sup_{\Probc}\!\Pr\nolimits_{F,\phi,\Probc}\!(\!\hat M\!\notin\!\set{M,\reject}\!).
\label{eq:err}
\end{equation}
A rate $R$ is \emph{achievable} if there exist $(2^{nR},n)$ codes with
$\varepsilon\to 0$, and $C^{\bullet}_{\mathrm{auth},\circ}$ denotes the
resulting capacity for $\bullet\!\in\!\set{\mathrm{fb},\mathrm{no\text{-}fb}}$,
$\circ\!\in\!\set{\mathrm{det},\mathrm{stoch}}$. We write
$\epsilon_n\!\triangleq\!\varepsilon(F,\phi)$; since
$\epsilon_n$ upper-bounds the no-adversary error, any bound stated in
terms of the latter also holds with $\epsilon_n$.

\section{Main results}
\label{sec:main}
\subsection{Overwritability conditions}
\label{sec:main:OW}
Our analysis uses two conditions strictly weaker than oblivious
overwritability~\eqref{eq:obow_def}. Each says the adversary can pick
states so the attacked channel matches the no-adversary channel
$W(\cdot\!\mid\!x',\szero)$, at the indicated level of granularity.

\begin{definition}[Distribution-overwritable]
\label{def:distow}
$\chn$ is \emph{distribution-overwritable} if for every
$\Prob\!\in\!\mathcal{P}(\cX)$ and every $x'\!\in\!\cX$ there is
$\Probc\!\in\!\mathcal{P}(\cS)$ such that
\begin{equation}
\!\!\!\sum_{x,s}\Prob(x)\Probc(s)\,W(y\!\mid\!x,s)=W(y\!\mid\!x',\szero)\quad\forall y\!\in\!\cY.
\label{eq:distow_def}
\end{equation}
\end{definition}
Definition~\ref{def:distow} coincides with the single-letter form of
the \emph{myopic distribution-overwritability}
of~\cite{BakshiK:23ISIT} when the adversary's side channel is trivial;
the general myopic case has additional
subtleties~\cite{BakshiK:23ISIT,BakshiBeemerK:prep}.

\begin{definition}[I-overwritable~\cite{BeemerGKKY20mildly}]
\label{def:Iow}
$\chn$ is \emph{I-overwritable} if for every $x,x'\!\in\!\cX$ there is
$\Probc\!\in\!\mathcal{P}(\cS)$ such that
\begin{equation}
\sum_{s\in\cS}\Probc(s)\,W(y\!\mid\!x,s)=W(y\!\mid\!x',\szero)\quad\forall y\!\in\!\cY.
\label{eq:Iow_def}
\end{equation}
\end{definition}

\begin{remark}[Strict hierarchy]
\label{rem:nest}
Obliviously overwritable $\Rightarrow$ distribution-overwritable
$\Rightarrow$ I-overwritable. As noted in~\cite{BakshiK:23ISIT,BakshiBeemerK:prep} both implications are strict: a
bit-flip channel is I-overwritable but not distribution-overwritable\ifarxiv\ (Appendix~\ref{app:ex:bitflip})\fi,
and the AZBSC$(\delta,\alpha)$ of~\cite{BakshiBeemerK:prep} is
distribution-overwritable but not obliviously overwritable\ifarxiv\ (Appendix~\ref{app:azbsc})\fi.
\ifarxiv\else For completeness, we reprise the details in~\extapps.\fi
\end{remark}

\begin{remark}[Non-overwritability witness]
\label{rem:nonow}
Since the alphabets are finite, the involved sets of distributions are
compact. Hence, when $\chn$ is not distribution-overwritable there
exist $\Prob_X\!\in\!\mathcal{P}(\cX)$, $x'\!\in\!\cX$, and
$\delta\!>\!0$ with
\begin{align}
\min_{\Probc\!\in\!\mathcal{P}(\cS)}\,
\mathbb{V}\bigl(\textstyle\sum_{x,s}\Prob_X(x)\Probc(s)W(\cdot\!\mid\!x,s),W(\cdot\!\mid\!x',\!\szero)\bigr)\!\geq\!\delta;
\label{eq:nonow}
\end{align}
similarly, when $\chn$ is not I-overwritable there exist
$x_0,x_0'\!\in\!\cX$ and $\delta\!>\!0$ with
\begin{align}
\min_{\Probc\!\in\!\mathcal{P}(\cS)}\!
\mathbb{V}\bigl(\textstyle\sum_{s}\Probc(s)W(\cdot\!\mid\!x_0,s),W(\cdot\!\mid\!x_0',\szero)\bigr)\!\geq\!\delta.
\label{eq:Iow_separation}
\end{align}
We call $(\Prob_X,x',\delta)$ the \emph{distribution-overwritability
witness} and $(x_0,x_0',\delta)$ the \emph{I-overwritability witness};
both are used by the achievability constructions in
Sections~\ref{sec:positivity}--\ref{sec:capacity}.
\end{remark}

\subsection{Authentication capacity and positivity conditions}

\begin{theorem}[Positivity for stochastic codes]
\label{thm:posstoch}
The following are equivalent:
\begin{enumerate}[(i)]
\item $C^{\mathrm{no\text{-}fb}}_{\mathrm{auth,stoch}}>0$;
\item $C^{\mathrm{fb}}_{\mathrm{auth,stoch}}>0$;
\item $\chn$ is not distribution-overwritable.
\end{enumerate}
\end{theorem}

\begin{theorem}[Stochastic capacity characterization]
\label{thm:caprate}
Whenever $C^{\bullet}_{\mathrm{auth,stoch}}\!>\!0$,
\begin{equation}
C^{\bullet}_{\mathrm{auth,stoch}}=C_{\szero}\quad\text{for }\bullet\!\in\!\set{\mathrm{fb},\mathrm{no\text{-}fb}}.
\end{equation}
\end{theorem}

\begin{theorem}[Positivity for deterministic codes]
\label{thm:posdet}
The following are equivalent:
\begin{enumerate}[(i)]
\item $C^{\mathrm{no\text{-}fb}}_{\mathrm{auth,det}}>0$;
\item $C^{\mathrm{fb}}_{\mathrm{auth,det}}>0$;
\item $\chn$ is not I-overwritable.
\end{enumerate}
\end{theorem}

\begin{remark}
\label{rem:detopen}
For deterministic codes our results give only positivity. The achievable
rate is not known to equal $C_{\szero}$: after the wait phase, the
adversary can guess the intended codeword and act effectively
\emph{omniscient}, but the omniscient-adversary capacity for
deterministic codes is itself open. Determining
$C^{\bullet}_{\mathrm{auth,det}}$ is left as an open problem. This mirrors a similar open question in the AVC reliability literature, where the capacity of deterministic codes under omniscient adversaries is not known in general~\cite{Langberg08,DeyJL15,ChenJL19}.
\end{remark}

\section{The wait-and-overwrite attack}
\label{sec:converse}

We prove the converse halves of Theorems~\ref{thm:posstoch}
and~\ref{thm:posdet} in this section via a common analysis. The
analysis builds on a wait-and-overwrite attack
(Algorithm~\ref{alg:attack}) that uses the \emph{posterior guessing
probability} to switch between a wait phase and an overwrite phase.

\begin{definition}[Posterior guessing probability]
\label{def:guess}
The posterior guessing probability for a random variable $A\!\in\!\cA$
given $B\!\in\!\cB$ is
\begin{equation}
\Probg_{A\mid B}(b)\triangleq\sum_{a\in\cA}\bigl(\Pr_{A\mid B}(a\!\mid\!b)\bigr)^2.
\label{eq:guess_def}
\end{equation}
\end{definition}
\begin{algorithm}[b]
\SetKwInOut{KwIn}{Input}\KwIn{Encoder $\set{\Prob_{X_t\mid \condmsg}}$, threshold $\gamma$, observations $y_1,\ldots,y_n$.}
\tcc{Wait phase}
$t\gets 1$\;
\While{$\Probg_{M\mid Y^{t-1}\!,S^{t-1}}(y^{t-1}\!,\szero^{t-1})\!\leq\!\gamma$ and $t\!\leq\!n$}{$s_t\gets\szero$;\ $t\gets t+1$\;}
$\Tg\gets t-1$; sample $\hat m,\tilde m\stackrel{\text{i.i.d.}}{\sim}\posnoadv\!(\cdot\!\mid\!y^{\Tg},\szero^{\Tg})$\;
\tcc{Overwrite phase}
\For{$\tau=\Tg+1,\ldots,n$}{Sample $s_\tau\!\sim\!\Probc^{(\tau)}$ so that the marginal at step $\tau$ matches a no-adv transmission of $\tilde m$\;}
\caption{Wait-and-overwrite attack.}
\label{alg:attack}
\end{algorithm}
\subsection{The attack}
\label{ssec:attack}

Fix a code $(F,\phi)$ and a threshold $\gamma\!\in\!(0,1/(32|\cY|^2))$.
During the wait phase the adversary sets
$S_t\!=\!\szero$ until the no-adversary posterior
$\Probg_{M\mid Y^{t-1}\!,S^{t-1}}$ exceeds $\gamma$; let $\Tg$ be the
first such index. The adversary then samples a guess $\hat M$ and a
decoy $\tilde M$ i.i.d.\ from the posterior at $\Tg$ and, in the
overwrite phase, plays a per-step strategy $\Probc^{(\tau)}$ so
the marginal law of $Y_\tau$ matches a no-adversary transmission
of $\tilde M$ under the encoder's use of $\hat M$.
The pseudocode is given in Algorithm~\ref{alg:attack}.

\subsection{Stopping-time lemma}
\label{ssec:stopping}

The wait phase rests on the following stopping-time control of the
posterior. Let $M$ be uniform on $[N]$ and $Y^n$ be the channel outputs
under $S^n\!=\!\szero^n$.

\begin{lemma}[Stopping time]
\label{lem:stopping}
Fix $\gamma\!\in\!(0,1/(32\lvert\cY\rvert^2))$ and let
\begin{equation}
\Tg\,\triangleq\,\inf\set{t\!\in\![n]:\PgM(Y^t\!,\szero^t)>\gamma},
\end{equation}
with the convention $\inf\emptyset\!=\!\infty$. For any
$(N,n)$ stochastic feedback code with no-adversary error probability
$p_e$,
\begin{align}
&\Pr(\Tg\!>\!n\mid S^n\!=\!\szero^n)\leq\frac{p_e}{1-\sqrt\gamma},\label{eq:stopping1}\\
&\Pr\!\bigl(\Probg_{M\mid Y^{\Tg}\!,S^{\Tg}}\!>\!\tfrac12\bigm\vert\Tg\!\leq\!n,\,S^{\Tg}\!=\!\szero^{\Tg}\bigr)\leq\sqrt 2\,\lvert\cY\rvert\sqrt\gamma\!<\! 1/4.
\label{eq:stopping2}
\end{align}
\end{lemma}

Equation~\eqref{eq:stopping1} says reliability of the code
($p_e\!\to\!0$) forces $\Tg\!\leq\!n$ with high probability.
Equation~\eqref{eq:stopping2} says that on the stopping event the
posterior at $\Tg$ stays below $1/2$ with probability $\geq\!3/4$,
which is the non-degeneracy that yields both $\hat M\!\neq\!M$ and
$\tilde M\!\neq\!M$ with constant probability.

\parnoindent{Proof sketch}
For~\eqref{eq:stopping1}, $g(y^n)\!\geq\!(1\!-\!p^*(y^n))^2$ where
$p^*$ is the conditional MAP error, so $g(Y^n)\!\leq\!\gamma$ forces
$p^*(Y^n)\!\geq\!1\!-\!\sqrt\gamma$, and Markov on
$\Exp[p^*(Y^n)]\!\leq\!p_e$ gives the bound.
For~\eqref{eq:stopping2}, on $\set{\Tg\!=\!t}$ minimality gives
$g(Y^{t-1})\!\leq\!\gamma$, and a sub-multiplicative bound on
$\sqrt{\Probg}$ (\extref{app:guessing}) combined with Markov yields
the displayed inequality. Full proof in \extref{app:stopping}.

\subsection{Proof of converses}
\label{ssec:proofconverse}

The wait-and-overwrite attack (Algorithm~\ref{alg:attack}) runs the
same way in both the stochastic and deterministic cases. The
adversary samples
$\hat m,\tilde m$ i.i.d.\ from the posterior at $\Tg$, commits to
the hypothesis $M\!=\!\hat m$, and at each step $\tau\!>\!\Tg$
chooses the per-step state distribution $\Probc^{(\tau)}$ as a
function of the realized history $(\hat m,\tilde m,y^{\tau-1},s^{\tau-1})$
so that the conditional law of $Y_\tau$ under the attack matches
the no-adversary conditional law of $Y_\tau$ under $M\!=\!\tilde m$
given the same past:
\begin{enumerate}[itemsep=0pt, label=--, wide, labelindent=0pt]
\item under distribution-overwritability (stochastic codes), let
$\Proba^{(\tau)}$ be the encoder's conditional input distribution
at step $\tau$ given $\bigl(M\!=\!\hat m,y^{\tau-1},s^{\tau-1}\bigr)$
(computed by the path likelihood over $x^{\tau-1}$), and let
$\Prob^{(\tau)}$ be the encoder's no-adversary conditional input
distribution at step $\tau$ given $\bigl(M\!=\!\tilde m,y^{\tau-1},\szero^{\tau-1}\bigr)$.
By distribution-overwritability, there is a $\Probc^{(\tau)}$ with
 \begin{align}
    &\sum_{x,s} \Proba^{(\tau)}(x) \Probc^{(\tau)}(s) \chn(\cdot \mid x,s) \nonumber \\
    &\qquad = \sum_{x'} \Prob^{(\tau)}(x') \chn(\cdot \mid x', \szero).
\end{align}
The precise per-step construction of $\Probc^{(\tau)}$ is detailed in~\extref{app:converse}.
\item under I-overwritability (deterministic codes), the codeword
symbols $c_{\hat m,\tau},c_{\tilde m,\tau}$ are determined by the
respective messages and the realized $y^{\tau-1}$ (under encoder
feedback), and the adversary's $\Probc^{(\tau)}$ satisfies the per-symbol identity
\[\sum_s\Probc^{(\tau)}(s)W(\cdot\!\mid\!c_{\hat m,\tau},s)=W(\cdot\!\mid\!c_{\tilde m,\tau},\szero);\]
\end{enumerate}
The success event for the attack is therefore
$\set{\hat m\!=\!M,\,\tilde m\!\neq\!M}$.

By Lemma~\ref{lem:stopping}, under $\mathcal{E}\!\triangleq\!\set{\Tg\!\leq\!n,\,\Probg_{M\mid Y^{\Tg}}\!\leq\!1/2}$,
we have $\Pr(\mathcal{E})\!\geq\!3/4\!-\!\epsilon_n/(1\!-\!\sqrt\gamma)$.
Since $\hat m,\tilde m,M$ are conditionally i.i.d.\ from the posterior
given $Y^{\Tg}$, the success event has conditional probability
$\Probg(1\!-\!\Probg)$, which is at least $\gamma(1\!-\!\gamma)$ under 
$\mathcal{E}$. Conditioned on the success event, the channel output is
a no-adversary transcript of $\tilde m\!\neq\!M$, so reliability of
the code makes the decoder output $\tilde m$ with probability
$\geq\!1\!-\!\epsilon_n$ -- an authentication error. Multiplying
gives the bound:
\begin{equation}
\Pr_{\PADV}\!\bigl(\hat M\!\notin\!\set{M,\reject}\bigr)\geq
\gamma(1\!-\!\gamma)(1\!-\!\epsilon_n)\!\left[\tfrac34-\tfrac{\epsilon_n}{1\!-\!\sqrt\gamma}\right].
\label{eq:converse_unified}
\end{equation}

\begin{proposition}[Converse halves of Theorems~\ref{thm:posstoch} and~\ref{thm:posdet}]
\label{thm:converse}
Fix $\gamma\!\in\!(0,1/(32\lvert\cY\rvert^2))$. For any blocklength-$n$
feedback code $(F,\phi)$ with auth-error $\epsilon_n$, the
wait-and-overwrite attack
satisfies~\eqref{eq:converse_unified} provided either
\begin{enumerate}[(i)]
\item $(F,\phi)$ is a stochastic code and $\chn$ is
distribution-overwritable, or
\item $(F,\phi)$ is a deterministic code and $\chn$ is
I-overwritable.
\end{enumerate}
In particular, $C^{\mathrm{fb}}_{\mathrm{auth,stoch}}\!=\!0$ in case (i)
and $C^{\mathrm{fb}}_{\mathrm{auth,det}}\!=\!0$ in case (ii).
\end{proposition}

\begin{remark}
\label{rem:bounds_same}
The two overwritability conditions affect \emph{which} channels admit
the attack, not the success probability. The bottleneck in both cases
is the event $\set{\hat m\!=\!M}$, of conditional probability
$\geq\!\gamma$ on $\mathcal{E}$, giving the $\gamma(1\!-\!\gamma)$
factor in~\eqref{eq:converse_unified}.
\end{remark}

\begin{remark}[Encoder feedback does not change the positivity]
Proposition~\ref{thm:converse} uses only the encoder's input
\emph{distribution} under feedback, not the absence of encoder feedback. The
same attack applies verbatim with or without encoder feedback, so both
cases of the proposition cover $C^{\mathrm{fb}}$ and
$C^{\mathrm{no\text{-}fb}}$.
\end{remark}

\section{Positivity proofs for Theorems~\ref{thm:posstoch} and~\ref{thm:posdet}}
\label{sec:positivity}

We sketch the achievability of Theorem~\ref{thm:posstoch};
Theorem~\ref{thm:posdet} follows along similar lines.
Fix $(\lambda_1,\lambda_2)\!\in\!(0,1)^2$. Our construction adapts the
overlapping-sets scheme of~\cite{BakshiK:23ISIT} to the
output-feedback adversary via a flip-set martingale concentration. We
build a code with rate
$({1}/{n})\log N_n^{\mathrm{(auth)}}\!\geq\!R^{\mathrm{(auth)}}\!>\!0$,
no-adversary error $\leq\!\lambda_1$, and missed-detection error
$\leq\!\lambda_2$ uniformly over feedback adversaries $\Probc$, for all
$n$ large enough:

\parnoindent{\underline{Step 1: Codebook}}\quad
Fix $\alpha\!\in\!(0,1)$ and a slack $\beta\!\in\!(0,1)$. By~\cite[Lemma~3]{BakshiK:23ISIT},
build a family $\mathfrak{B}\!=\!\set{\cB_m\!\subset\![n]}_{m\in[N]}$
of subsets of $[n]$ with $N\!\geq\!2^{R^{\mathrm{(auth)}}n}$, in
which each set has size in $\alpha(1\!\pm\!\beta)n$, distinct sets
have intersection at most $\alpha^2(1\!+\!\beta)n$, and an
\emph{overlap property} that we use in Step~3.

\parnoindent{\underline{Step 2: Encoder/Decoder}}\quad Pick a non-overwritability
witness $(\Prob_X,x',\delta)$ from~\eqref{eq:nonow}. The encoder
transmits $X_t\!=\!x'$ on $\cB_m$ and $X_t\!\sim\!\Prob_X$ i.i.d.\ off
$\cB_m$. Given $Y^n$, for any $\cB\!\subseteq\![n]$ let
$\widehat P_{Y^n}(\cdot\!\mid\!\cB)$ denote the empirical type of
$Y^n$ restricted to positions $t\!\in\!\cB$, i.e.\
$\widehat P_{Y^n}(y\!\mid\!\cB)\!=\!\frac{1}{|\cB|}\sum_{t\in\cB}\mathbf{1}\set{Y_t\!=\!y}$.
The decoder computes
\begin{equation}
T_{\hat m}(Y^n)\!\triangleq\!\mathbb{V}\!\bigl(\widehat P_{Y^n}(\cdot\!\mid\!\cB_{\hat m}),W(\cdot\!\mid\!x',\szero)\bigr)
\label{eq:authcode:test_sketch}
\end{equation}
for each $\hat m\!\in\![N]$ and outputs the unique $\hat m$ with
$T_{\hat m}(Y^n)\!\leq\!\delta/4$, or $\reject$ when none or
several candidates pass.

\parnoindent{\underline{Step 3: Decoding-error bound}}\quad The decoder errs only
when some wrong $\hat m\!\neq\!M$ also passes the test. We show
that for each fixed $m\!\neq\!\hat m$ and any feedback adversary
$\Probc$,
\begin{equation}
T_{\hat m}(Y^n)\!\geq\!\delta/2\quad\text{with probability }1\!-\!e^{-\Omega(n)},
\label{eq:authcode:bad_target}
\end{equation}
and a union bound over $N\!\leq\!2^{R^{(\mathrm{auth})}n}$
candidates closes the argument once $R^{(\mathrm{auth})}$ is
chosen small enough. The overlap property of Step~1 reduces
\eqref{eq:authcode:bad_target} to a TV bound on the flip set
$J\!\triangleq\!\cB_{\hat m}\!\setminus\!\cB_m$:
\begin{equation}
\mathbb{V}\!\bigl(\widehat P_{Y^n}(\cdot\!\mid\!J),W(\cdot\!\mid\!x',\szero)\bigr)\!\geq\!\delta/2.
\label{eq:flip_target}
\end{equation}
A per-letter analysis cannot proceed directly: $S_t$ depends on $Y^{t-1}$,
so the $Y_t$'s on the flip set are neither independent nor identically
distributed, and the strategy $\set{\Probc_{S_t\mid\cF_{t-1}}}$ ranges
over an uncountable path-dependent class. Non-overwritability~\eqref{eq:nonow}
sidesteps this: for \emph{every} single-letter $Q\!\in\!\mathcal{P}(\cS)$,
the averaged channel $\sum_{x,s}\Prob_X(x)Q(s)W(\cdot\!\mid\!x,s)$ is
$\delta$-far in TV from $W(\cdot\!\mid\!x',\szero)$. On the flip set the
encoder draws fresh i.i.d.\ $\Prob_X$, so with
$\cF_{t-1}\!\triangleq\!\sigma(M,X^t,S^t,Y^t)$ and
$Q\!=\!\Probc_{S_t\mid\cF_{t-1}}$, the conditional law $\pi_t$ of $Y_t$
given $\cF_{t-1}$ satisfies
$\mathbb{V}(\pi_t,W(\cdot\!\mid\!x',\szero))\!\geq\!\delta$ on every
history. The centred sequence
$Z_t\!\triangleq\!\mathbf{1}\set{Y_t\!\in\!A}\!-\!\pi_t(A)$ is then a
bounded-difference $\set{\cF_t}$-martingale with zero conditional mean
regardless of $\Probc$, and Azuma--Hoeffding gives~\eqref{eq:flip_target}
with probability $1\!-\!e^{-\Omega(n)}$ and an exponent depending only on
$(\delta,\alpha,\beta,\lvert\cY\rvert)$. The supremum over $\Probc$ moves
outside the probability.

\subsection{Achievability for deterministic codes (Theorem~\ref{thm:posdet})}
\label{sec:detach}

\ifarxiv
We prove the achievability direction of Theorem~\ref{thm:posdet} as the
specialization $\Prob_X\!=\!\delta_{x_0}$ of the stochastic auth-code
construction of Section~\ref{sec:positivity}; if $\chn$ is not I-overwritable
then $C^{\bullet}_{\mathrm{auth,det}}\!>\!0$. Concretely, the
witness $(\Prob_X,x',\delta)$ becomes
$(\delta_{x_0},x_0',\delta)$, with $\delta$ now the
\emph{I-overwritability} gap of \eqref{eq:Iow_separation}. The key differences
from the stochastic case are in three places:
\begin{itemize}
\item \emph{Encoder.} With $\Prob_X\!=\!\delta_{x_0}$, the encoder is
fully deterministic: $X_t\!=\!x_0'$ for $t\!\in\!\cB_m$ and
$X_t\!=\!x_0$ for $t\!\notin\!\cB_m$. The codeword is binary,
$c_m\!\in\!\set{x_0,x_0'}^n$.
\item \emph{Witness gap.} The dist-OW gap in step~(1) of the auth-tag
proof reduces, since $\Prob_X$ is a Dirac mass, to the I-OW gap
\eqref{eq:Iow_separation}, namely
\[\inf_{\Probc}\mathbb{V}\!\bigl(\textstyle\sum_s\Probc(s)W(\cdot\!\mid\!x_0,s),W(\cdot\!\mid\!x_0',\szero)\bigr)\!\geq\!\delta.\]
\item \emph{Step~(2) martingale.} On the flip set the encoder
transmits $x_0$ deterministically; the conditional Bernoulli mean is
$p_t\!=\!\sum_s\Probc_{S_t\mid\cF_{t-1}}(s\mid\cF_{t-1})W(A\!\mid\!x_0,s)$,
$\delta$-far from $W(A\!\mid\!x_0',\szero)$ uniformly over
$\Probc$. Azuma--Hoeffding on the same bounded-difference martingale
gives the same flip-set concentration.
\end{itemize}
The codebook (overlapping-sets lemma) and the overlap property
(step~(3)) are unchanged -- they are statements about set geometry
and do not depend on the input distribution. Combining gives
$C^{\bullet}_{\mathrm{auth,det}}\!\geq\!R\!>\!0$ for both
$\bullet\!\in\!\set{\mathrm{fb},\mathrm{no\text{-}fb}}$. The full
deterministic proof is in \extref{app:detach_full}.

\begin{remark}
\label{rem:detach_vs_stoch}
Let $C^{*}\!\triangleq\!\max_{\Prob_X\in\mathcal{P}(\set{x_0,x_0'})}
I(X;Y\!\mid\!S\!=\!\szero)$ be the Shannon capacity of the
no-adversary sub-channel $W(y\!\mid\!x,\szero)$ under inputs
restricted to $\set{x_0,x_0'}$. The proof yields a positive rate
strictly below $C^{*}$. The gap to $C^{*}$ (and a fortiori to
$C_{\szero}$) is an open problem -- see Remark~\ref{rem:detopen}.
\end{remark}
\else
The achievability direction of Theorem~\ref{thm:posdet} follows the
same template as the auth-code construction of
Section~\ref{sec:positivity}, specialized to the Dirac input
distribution $\Prob_X\!=\!\delta_{x_0}$. The dist-OW gap is then replaced
by the I-OW gap~\eqref{eq:Iow_separation}, and the flip-set
martingale concentration carries over verbatim. The codebook and the
overlap property are independent of the input distribution. Full
details are in~\extref{app:detach_full}.
\fi

\section{Capacity for stochastic codes (Theorem~\ref{thm:caprate})}
\label{sec:capacity}

The auth-code of Section~\ref{sec:positivity} only guarantees an
exponential number of messages, short of $C_{\szero}$. We reach the
no-adversary capacity in two steps: (i) build an
\emph{authentication tag} with doubly-exponentially many candidates by
composing the auth-code with an Ahlswede--Dueck identification code
for the noiseless binary channel;
(ii) concatenate the tag with a capacity-achieving channel code on
$W_{Y\mid X,S=\szero}$, with the channel code carrying $M$ in $n_1$
symbols and the tag carrying a verification statistic in
$n_2\!=\!O(\log n)$ symbols.

\subsection{Authentication tags}
\label{ssec:authtag}
An $(N,n)$-\emph{authentication tag}~\cite[Def.~9]{BakshiK:23ISIT}
consists of a stochastic encoder $\Prob_{X^n\mid M}$ and a family of
deterministic verifiers
$\set{\authdec_{\hat m}\colon\cY^n\!\to\!\set{\accept,\reject}}_{\hat m}$
indexed by a candidate message $\hat m$ already available to the
receiver, with low false-alarm and missed-detection error (formal
definition in \extref{app:authtag}). The tag's \emph{rate} is the
doubly-exponential growth $\frac{1}{n}\log\log N_n$.

\begin{lemma}[Existence of authentication tags]
\label{lem:tag}
If $\chn$ is not distribution-overwritable, then there exists
$R\!>\!0$ such that, for every $(\lambda_1,\lambda_2)\!\in\!(0,1)^2$
and all $n$ sufficiently large, there is an $(N_n,n)$-tag with
errors $(\lambda_1,\lambda_2)$ and
$\liminf_{n\to\infty}\frac{1}{n}\log\log N_n\!\geq\!R$.
\end{lemma}

The proof has two parts. First, we take an authentication code of
Section~\ref{sec:positivity}, whose missed-detection bound is
uniform over output-feedback adversaries. Next, we compose
this auth-code with an Ahlswede--Dueck identification
code~\cite{AhlswedeD89} for the noiseless binary channel: each
auth-code codeword is reinterpreted as the output of a length-$\tilde
n$ identification code with $\lfloor 2^{2^{\tilde n\tilde R}}\rfloor$
candidate messages, and the verifier on a candidate $\hat m$ first
runs the auth-code's decoder and then the identification verifier
on the decoded codeword. The result is a tag with
$\frac{1}{n}\log\log N_n\!\geq\!\tilde R\hat R\!=\!R$.

\subsection{Concatenation with a channel code}
\label{ssec:concatenation}
\begin{figure}[t]
\centering
\resizebox{\columnwidth}{!}{%
\begin{tikzpicture}[
  >=stealth, font=\scriptsize,
  block/.style={draw, rounded corners=1.5pt, minimum height=18pt,
                minimum width=30pt, align=center, inner sep=2pt,
                font=\scriptsize},
  io/.style={font=\scriptsize}
]
% Source
\node[io] (M) at (0,0) {$M$};

% Encoders
\node[block, fill=blue!8] (Fch)  at (1.1, 0.55) {$\Prob^{\mathrm{ch}}$};
\node[block, fill=blue!8] (Ftag) at (1.1,-0.55) {$\Prob^{\mathrm{tag}}$};
\draw[->] (M) -- (Fch);
\draw[->] (M) -- (Ftag);

% Channel
\node[block, fill=red!8, minimum height=34pt, minimum width=22pt] (chan) at (3.2,0) {$\chn$};
\draw[->] (Fch.east)  -- node[above,font=\tiny]{$X^{n_1}$}        (chan.west |- Fch.east);
\draw[->] (Ftag.east) -- node[below,font=\tiny]{$X_{n_1+1}^{n_1+n_2}$} (chan.west |- Ftag.east);

% Adversary selects channel state from past observations Y^{t-1}.
\node[block, fill=red!4, draw=red!50!black, font=\tiny] (adv) at (3.2, 1.25) {Adversary};
\draw[->, red!60!black] (adv) -- node[right,font=\tiny,red!60!black]{$S_t$} (chan.north);
\draw[->, red!60!black, dashed] (4.2,1.25) -- node[right,font=\tiny,red!60!black]{$Y^{t-1}$} (adv.east);

% Decoders
\node[block, fill=green!8] (phich) at (5.3, 0.55) {$\phi^{\mathrm{ch}}$};
\node[block, fill=green!8] (Phitag) at (5.3,-0.55) {$D^{\mathrm{tag}}_{\hat m}$};
\draw[->] (chan.east |- phich) -- node[above,font=\tiny]{$Y^{n_1}$}        (phich.west);
\draw[->] (chan.east |- Phitag) -- node[below,font=\tiny]{$Y_{n_1+1}^{n_1+n_2}$} (Phitag.west);

% Output
\node[io] (out) at (7.6,0) {$\hat M\!\in\![N]\!\cup\!\set{\reject}$};
\draw[->] (phich) -- node[above,font=\tiny]{$\hat m$} (out);
\draw[->] (Phitag) -- node[below,font=\tiny,align=center]
  {\ $\accept$\\ \ or\\ \ $\reject$} (out);
\draw[->] (phich.south) -- node[right,font=\tiny]{$\hat m$} (Phitag.north);
\end{tikzpicture}}
\caption{Code construction for Theorem~\ref{thm:caprate}.}
\label{fig:achievability}
\end{figure}
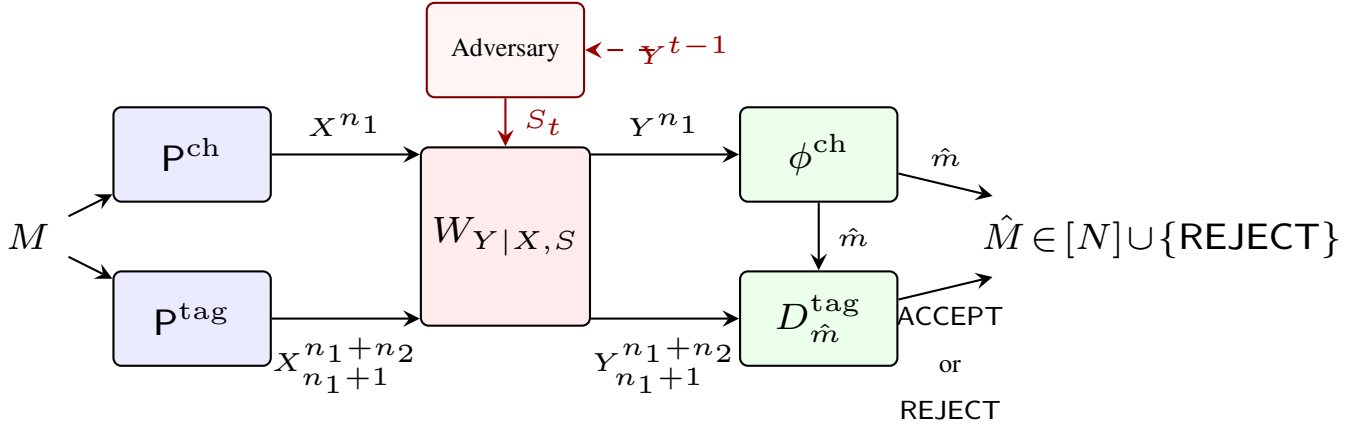

%\parnoindent{Code construction}
The achievability scheme concatenates a capacity-achieving channel
code with the tag of Lemma~\ref{lem:tag} (Fig.~\ref{fig:achievability}). Fix $R\!<\!C_{\szero}$. Choose blocklength $n$ so that:
\begin{enumerate}[itemsep=0pt, label=(\emph{\roman*}), wide, labelindent=0pt]
\item there is an $(\lceil 2^{nR}\rceil,n_1)$ channel code
$(\Prob^{\mathrm{ch}}_{X^{n_1}\mid M},\phi^{\mathrm{ch}})$ over $W_{Y\mid X,S=\szero}$ with average error
$\epsilon_n\!=\!o(1)$;
\item there is an $(\lceil 2^{nR}\rceil,n_2)$ authentication tag
$(\Prob^{\mathrm{tag}}_{X^{n_2}\mid M},\set{\authdec^{\mathrm{tag}}_{\hat m}})$
of error pair $(\epsilon_n,\epsilon_n)$ for $\chn$, guaranteed by
Lemma~\ref{lem:tag} as long as
$n_2\!=\!\Omega(\log\log\lceil 2^{nR}\rceil)\!=\!O(\log n)$.
\end{enumerate}
The composite encoder uses
\begin{equation}
\Prob_{X^{n_1+n_2}\!\mid\! M}(x^{n_1+n_2}\!\!\mid\! m)
\!=\!\Prob^{\mathrm{ch}}_{X^{n_1}\!\mid M}\!(x^{n_1}\!\mid\! m)\Prob^{\mathrm{tag}}_{X^{n_2}\mid\! M}\!(x_{n_1+1}^{n_1+n_2}\!\!\mid\!m).
\end{equation}
The decoder first runs $\phi^{\mathrm{ch}}$ on $Y^{n_1}$ to obtain a
candidate $\hat m\!=\!\phi^{\mathrm{ch}}(Y^{n_1})$, then verifies it via
the tag:
\begin{equation}
\phi(Y^{n_1+n_2})=\begin{cases}\hat m & \text{if }\authdec^{\mathrm{tag}}_{\hat m}(Y_{n_1+1}^{n_1+n_2})\!=\!\accept,\\
\reject & \text{otherwise.}\end{cases}
\label{eq:composite}
\end{equation}

\parnoindent{Error analysis} Under no adversary, the channel-code stage outputs $\hat m\!=\!M$ with
probability $\geq\!1\!-\!\epsilon_n$, and the verifier accepts the correct
tag with probability $\geq\!1\!-\!\epsilon_n$. So
$\Pr_{\PNA}(\phi\!\neq\!M)\!\leq\!2\epsilon_n$.

Under any adversary strategy, the channel-code stage outputs
\emph{some} candidate $\hat m\!\in\![N]$. If $\hat m\!=\!M$, the tag
verifier accepts and authentication succeeds. If
$\authdec^{\mathrm{tag}}_{\hat m}(Y_{n_1+1}^{n_1+n_2})\!=\!\reject$, the
composite decoder outputs $\reject$ and authentication succeeds.
Otherwise $\hat m\!\neq\!M$ \emph{and} the verifier accepts -- but this
is exactly the missed-detection event of the tag, whose probability is
at most $\epsilon_n$ for every adversary strategy. Combining,
\begin{equation}
\Pr_{\PADV}(\phi(Y^{n_1+n_2})\notin\set{M,\reject})\leq \epsilon_n.
\end{equation}
Hence the composite code achieves authentication error $\leq\!3\epsilon_n\!=\!o(1)$
at rate $\tfrac{nR}{n+O(\log n)}\!\to\!R$ as $n\!\to\!\infty$.

Since $R\!<\!C_{\szero}$ was arbitrary,
$C^{\mathrm{fb}}_{\mathrm{auth,stoch}}\!\geq\!C_{\szero}$, and combined
with the trivial reverse inequality this proves Theorem~\ref{thm:caprate}.

\begin{remark}During the wait phase the adversary may decode $M$ from
$Y^{n_1}$ and become effectively $M$-omniscient in the tag phase. The
tag's missed-detection bound nevertheless holds uniformly: the
encoder's fresh i.i.d.\ $\Prob_X$ on the flip set
$\cB_{\hat m}\!\setminus\!\cB_m$ hides $X_t$ at each position, and the
non-distribution-overwritability gap~\eqref{eq:nonow} controls $\pi_t$
uniformly over $\Probc_{S_t\mid\cF_{t-1}}$. Details are
in~\extref{app:authtag:composition}.\end{remark}

\section{Conclusion}
\label{sec:conclusion}

We examined the authentication capacity of an AVC against a
causal adversary with output feedback. For stochastic codes,
positivity is equivalent to $\chn$ not being
distribution-overwritable, and whenever the capacity is positive it
equals the no-adversary Shannon capacity $C_{\szero}$. For
deterministic codes, positivity is equivalent to $\chn$ not being
I-overwritable; characterizing the deterministic-code capacity remains open, paralleling the analogous gap in the
omniscient-adversary setting. The converses follow from a
wait-and-overwrite attack based on by the posterior guessing
probability, and the achievability uses an overlapping-set code
analyzed via a martingale that absorbs the adversary's
adaptivity. Reaching
$C_{\szero}$ requires composing this authentication code with an
Ahlswede--Dueck identification tag, to obtain an authentication tag and,  concatenating the resulting authentication tag with a
channel code. Natural directions for future work
include extending the framework to continuous alphabets, to feedback
adversaries with delayed or noisy observations, and to multi-user
settings.

\bibliographystyle{IEEEtran}
\bibliography{refs}

% Appendices: only included in the arXiv extended version.
\ifarxiv
\newpage
\appendices

\begin{table}[h]
\centering
\caption{Proof roadmap. Building blocks for the converse and
achievability are listed first; the main theorems combine the two
sides.}
\label{tab:roadmap}
\renewcommand{\arraystretch}{1.25}
\setlength{\tabcolsep}{4pt}
\begin{tabular}{@{}p{0.30\linewidth} p{0.30\linewidth} p{0.32\linewidth}@{}}
\toprule
Result & Proved in & Uses \\
\midrule
\multicolumn{3}{@{}l}{\emph{Converse-side building blocks}} \\
\midrule
Lemma~\ref{lem:guessing} (sub-multiplicativity) &
App.~\ref{app:guessing} &
--- \\
Lemma~\ref{lem:stopping} (stopping time) &
App.~\ref{app:stopping} &
Lemma~\ref{lem:guessing} \\
Claim~\ref{claim:gamma} (per-step matching) &
App.~\ref{app:converse} &
dist-OW (Def.~\ref{def:distow}) / I-OW (Def.~\ref{def:Iow}) \\
Proposition~\ref{thm:converse} (converse halves) &
App.~\ref{app:converse} &
Lemma~\ref{lem:stopping}, Claim~\ref{claim:gamma} \\
\midrule
\multicolumn{3}{@{}l}{\emph{Achievability-side building blocks}} \\
\midrule
Auth-code construction (single-exp.) &
Section~\ref{sec:positivity}, App.~\ref{app:authtag:proof} &
non-distribution-OW gap~\eqref{eq:nonow}, flip-set martingale \\
Lemma~\ref{lem:tag} (auth-tag, doubly-exp.) &
App.~\ref{app:authtag:idcomp} &
auth-code construction, identification codes \\
Concatenation with channel code &
App.~\ref{app:authtag:composition} &
Lemma~\ref{lem:tag}, capacity-achieving channel code on $W_{Y\mid X,S=\szero}$ \\
\midrule
\multicolumn{3}{@{}l}{\emph{Main theorems (combining both sides)}} \\
\midrule
Theorem~\ref{thm:posstoch} (stochastic positivity) &
Achiev.: Section~\ref{sec:positivity};\newline
Conv.: App.~\ref{app:converse}, case~(i) &
auth-code construction, Proposition~\ref{thm:converse} \\
Theorem~\ref{thm:caprate} (stochastic capacity) &
Achiev.: Section~\ref{sec:capacity}, App.~\ref{app:authtag};\newline
Conv.: App.~\ref{app:converse} &
Lemma~\ref{lem:tag}, Proposition~\ref{thm:converse} \\
Theorem~\ref{thm:posdet} (deterministic positivity) &
Achiev.: Section~\ref{sec:detach}, App.~\ref{app:detach_full};\newline
Conv.: App.~\ref{app:converse}, case~(ii) &
Const.-comp.\ codebook + I-OW, Proposition~\ref{thm:converse} \\
\midrule
\multicolumn{3}{@{}l}{\emph{Illustrative examples (Remark~\ref{rem:nest})}} \\
\midrule
Bit-flip channel &
App.~\ref{app:ex:bitflip} &
I-overwritable but not distribution-overwritable \\
AZBSC$(\delta,\alpha)$ &
App.~\ref{app:azbsc} &
distribution-overwritable but not obliviously overwritable \\
\bottomrule
\end{tabular}
\end{table}

\section{Sub-multiplicative growth of the guessing probability}
\label{app:guessing}

\begin{lemma}[Guessing-probability sub-multiplicativity]
\label{lem:guessing}
Under the no-adversary law $\PNA$,
\begin{equation}
\Exp\!\bigl[\sqrt{\Probg_{M\mid Y^t\!,S^t}(Y^t\!,\szero^t)}\bigm\vert Y^{t-1}\!=\!y^{t-1}\bigr]
\leq \lvert\cY\rvert\sqrt{\Probg_{M\mid Y^{t-1}\!,S^{t-1}}(y^{t-1}\!,\szero^{t-1})}.
\end{equation}
\end{lemma}

\noindent\emph{Proof.}\quad
Write $g(y^t)\!\triangleq\!\Probg_{M\mid Y^t\!,S^t}(y^t,\szero^t)$ and
$p(y_t\!\mid\!y^{t-1})\!\triangleq\!\Prob_{Y_t\mid Y^{t-1}\!,S^t}(y_t\!\mid\!y^{t-1},\szero^t)$.
For every $(m,y_t)$,
\begin{equation}
\Prob_{M,Y_t\mid Y^{t-1}}(m,y_t\!\mid\!y^{t-1})\leq \Prob_{M\mid Y^{t-1}}(m\!\mid\!y^{t-1}).
\end{equation}
Squaring and summing over $m$ gives
\begin{equation}
\sum_m\Prob_{M,Y_t\mid Y^{t-1}}(m,y_t\!\mid\!y^{t-1})^2\leq g(y^{t-1}).
\label{eq:joint_le_marg}
\end{equation}
Therefore
\begin{align}
&\Exp[\sqrt{g(Y^{t})}\!\mid\!y^{t-1}]\eqbr=\!\sum_{y_t}\!p(y_t\!\mid\!y^{t-1})\sqrt{\textstyle\sum_m\Prob_{M\mid Y^t}(m\!\mid\!y^t)^2}\\ \nonumber & =\!\sum_{y_t}\!\sqrt{\textstyle\sum_m\!\bigl(p(y_t\!\mid\!y^{t-1})\Prob_{M\mid Y^t}(m\!\mid\!y^t)\bigr)^2}\eqbr=\!\sum_{y_t}\!\sqrt{\textstyle\sum_m\Prob_{M,Y_t\mid Y^{t-1}}(m,y_t\!\mid\!y^{t-1})^2}\eqbr\leq\!\sum_{y_t}\!\sqrt{g(y^{t-1})}=\lvert\cY\rvert\sqrt{g(y^{t-1})}.
\end{align}
\hfill$\square$

\section{Proof of Lemma~\ref{lem:stopping}}
\label{app:stopping}

Throughout this appendix all probabilities and expectations are taken
under the no-adversary law $\PNA$, i.e., the joint law of $(M,X^n,Y^n)$
when the adversary plays $S^n\!=\!\szero^n$. Write
$g(y^t)\!\triangleq\!\Probg_{M\mid Y^t\!,S^t}(y^t,\szero^t)$ and recall
$\Tg=\inf\set{t\!\in\![n]:g(Y^t)\!>\!\gamma}$.

\emph{Proof of \eqref{eq:stopping1}.}\quad
Let $p_e^{(n)}\!\triangleq\!\Pr_{\PNA}(\hat M\!\neq\!M)$, and write
$p^*(y^n)\!\triangleq\!1\!-\!\max_m\Prob_{M\mid Y^n}(m\!\mid\!y^n)$ for
the conditional MAP error given $Y^n\!=\!y^n$. Since the MAP rule is
optimal,
$\Exp[p^*(Y^n)]\!\leq\!p_e^{(n)}$.
The guessing probability dominates the squared MAP success
probability:
\begin{equation}
g(y^n)\!\geq\!\bigl(\max_m\Prob_{M\mid Y^n}(m\!\mid\!y^n)\bigr)^2\!=\!(1-p^*(y^n))^2,
\end{equation}
so $g(Y^n)\!\leq\!\gamma$ forces
$p^*(Y^n)\!\geq\!1\!-\!\sqrt\gamma$. By Markov,
\begin{equation}
\Pr(\Tg\!>\!n)\!\leq\!\Pr(g(Y^n)\!\leq\!\gamma)\!\leq\!\frac{p_e^{(n)}}{1-\sqrt\gamma}\!\xrightarrow[n\to\infty]{}\!0.
\label{eq:gtoone}
\end{equation}

\emph{Proof of \eqref{eq:stopping2}.}\quad
Under the event $\set{\Tg\!=\!t}$, the definition of $\Tg$ as the
\emph{first} crossing time gives $g(Y^{t-1})\!\leq\!\gamma$. Combining
with Lemma~\ref{lem:guessing},
\begin{equation}
\Exp\!\bigl[\sqrt{g(Y^t)}\bigm\vert\Tg\!=\!t\bigr]\!\leq\!\lvert\cY\rvert\sqrt{\gamma}.
\end{equation}
Markov's inequality with $\gamma\!<\!1/(32\lvert\cY\rvert^2)$ yields
\begin{align}
\Pr\!\bigl(g(Y^t)\!>\!\tfrac12\bigm\vert\Tg\!=\!t\bigr)
&\leq\sqrt 2\cdot\Exp\!\bigl[\sqrt{g(Y^t)}\bigm\vert\Tg\!=\!t\bigr]\nonumber\\
&\leq\sqrt 2\,\lvert\cY\rvert\sqrt\gamma\!<\!\tfrac14.
\end{align}
Averaging over $t\!\in\![n]$ weighted by $\Pr(\Tg\!=\!t)$ gives
$\Pr(g(Y^{\Tg})\!>\!1/2\!\mid\!\Tg\!\leq\!n)\!<\!1/4$, i.e.,
\eqref{eq:stopping2}.\hfill$\square$

\section{Detailed proof of the converse}
\label{app:converse}

We give a single detailed proof of the converse
Proposition~\ref{thm:converse}, covering both the
stochastic/distribution-overwritable and the
deterministic/I-overwritable cases. We first establish the
per-step matching identity (Claim~\ref{claim:gamma}); then write the
attack joint distribution explicitly; decompose the resulting
distribution into a ``matched'' part ($M\!=\!\hat M$) and a
``mismatched'' part ($M\!\neq\!\hat M$); bound the rejection
probability and the correct-decoding probability separately; and
finally combine via the high-probability event~$\mathcal{E}$.

Throughout, $(F,\phi)$ is a feedback code of blocklength $n$ with
auth-error $\epsilon_n$,
$\cD_\perp\!\triangleq\!\phi^{-1}(\reject)$, and the wait-and-overwrite
attack of Section~\ref{sec:converse} stops at $\Tg$ and overwrites with
$\tilde M$. We use $\Proba$ for the joint distribution under the
attack and $\Prob_{\PNA}$ for the no-adversary joint. The argument is
written for the stochastic case; the deterministic case is the
specialization in which $\Proba^{(\tau)}\!=\!\mathbf{1}_{c_{\hat m,\tau}}$
is a point mass and dist-OW reduces to I-OW (cf.\ end of this
appendix).

\parnoindent{Per-step matching}

\begin{claim}\label{claim:gamma}
Suppose $\chn$ is distribution-overwritable. Fix
$\tau\!\in\![\Tg+1\!:\!n]$ and a history
$h^{\tau-1}\!\triangleq\!(\hat m,\tilde m,y^{\tau-1},s^{\tau-1})$. Let
$\Proba^{(\tau)}$ be the encoder's marginal at step $\tau$ under the
attack with $M\!=\!\hat M\!=\!\hat m$, and let $\Prob^{(\tau)}$ be the
encoder's marginal at step $\tau$ under no-adversary with
$M\!=\!\tilde m$. Then there exists
$\Probc^{(\tau)}\!\in\!\mathcal{P}(\cS)$ with
\begin{align}
&\sum_{x_\tau\!,s_\tau}\Proba^{(\tau)}\!(x_\tau)\Probc^{(\tau)}\!(s_\tau)\chn(y_\tau\!\mid\!x_\tau,s_\tau)\eqbr=\sum_{x'_\tau}\Prob^{(\tau)}\!(x'_\tau)\chn(y_\tau\!\mid\!x'_\tau,\szero)
\label{eq:claimstep_app}
\end{align}
for every $y_\tau\!\in\!\cY$.
\end{claim}

\noindent\emph{Proof.}\quad
$\Proba^{(\tau)}$ is computable from the encoder's conditional law and
the realized $(\hat m,y^{\tau-1},s^{\tau-1})$ via the path likelihood
\begin{equation}
\encno\!(x_\tau\!\mid\!m,h^{\tau-1})
=\frac{\Lambda_\tau\!(m,x_\tau,h^{\tau-1})}{\sum_{x'_\tau}\Lambda_\tau\!(m,x'_\tau,h^{\tau-1})},
\label{eq:Pinherit}
\end{equation}
where
\begin{align}
&\Lambda_\tau(m,x_\tau,h^{\tau-1})\eqbr=\sum_{x^{\tau-1}}\biggl(\prod_{i=1}^{\tau-1}\Prob_{X_i\mid \condmsg}(x_i\!\mid\!m,x^{i-1}\!,y^{i-1})\eqbr\times\chn(y_i\!\mid\!x_i,s_i)\biggr)\eqbr\times\Prob_{X_\tau\mid \condmsgseq}(x_\tau\!\mid\!m,x^{\tau-1}\!,y^{\tau-1}).
\end{align}
Averaging the per-$x'$ witnesses of Definition~\ref{def:distow}
against $\Prob^{(\tau)}(x')$ with $\Prob\!=\!\Proba^{(\tau)}$ yields a
$\Probc^{(\tau)}$ satisfying~\eqref{eq:claimstep_app}.\hfill$\square$

\paragraph*{Telescoping}
For $i\!\in\![\Tg+1\!:\!n]$, define the hybrid trajectory
\begin{align}
&\Gamma_{i}(y_{\Tg+1}^n\!\mid\!\hat m,\tilde m,y^{\Tg})\triangleq A_i\cdot B_i,
\label{eq:Gamma_app}
\end{align}
where
\begin{align}
A_i&\triangleq\!\!\sum_{s_{\Tg+1}^{i}}\prod_{\tau=\Tg+1}^{i}\!\Proba_{Y_\tau\!,S_\tau\mid \condfull}\!(y_\tau,s_\tau\!\mid\!\dots),\\
B_i&\triangleq\!\!\prod_{\tau=i+1}^{n}\!\sum_{x_\tau}\!\encno(x_\tau\!\mid\!\tilde m,y^{\tau-1}\!,\szero^{\tau-1})\eqbr\times\chn(y_\tau\!\mid\!x_\tau,\szero).
\end{align}
$A_i$ records the actual attack up through step $i$, and $B_i$ the
no-adversary distribution thereafter. Substituting~\eqref{eq:claimstep_app}
into the inner sum over $s_i$ inside $A_i$ collapses the $\tau\!=\!i$
factor of $A_i$ onto the $\tau\!=\!i$ factor of $B_{i-1}$, so
$\Gamma_i\!=\!\Gamma_{i-1}$. Iterating from $i\!=\!n$ down to
$i\!=\!\Tg+2$ yields the telescoping identity
\begin{equation}
\Gamma_n=\Gamma_{\Tg+1}.
\label{eq:gamma_telescope_app}
\end{equation}

\parnoindent{Explicit attack joint distribution}
\label{app:long:attack}

For $t\!\in\![n]$ and
$y^n\!\in\!\cG_t\!=\!\set{\Tg(y^n)\!=\!t}$, the wait phase forces
$S^t\!=\!\szero^t$ and the overwrite phase samples $\hat M,\tilde M$
i.i.d.\ from $\posnoadv(\cdot\!\mid\!y^t,\szero^t)$ and uses the
strategies~$\set{\Probc^{(\tau)}}_{\tau>t}$. Hence
\begin{align}
&\Proba(M\!=\!m,\hat M\!=\!\hat m,\tilde M\!=\!\tilde m,Y^n\!=\!y^n,S^n\!=\!(\szero^t,s_{t+1}^n))\\
&\quad=\Prob_M(m)\,\Prob^{\PNA}_{Y^t\mid M,S^t}(y^t\!\mid\!m,\szero^t)\,\posnoadv(\hat m\!\mid\!y^t,\szero^t)\\
&\qquad\cdot\posnoadv(\tilde m\!\mid\!y^t,\szero^t)\\
&\qquad\cdot\Proba_{Y_{t+1}^n,S_{t+1}^n\mid M,\hat M,\tilde M,Y^t,S^t}(y_{t+1}^n,s_{t+1}^n\!\mid\!m,\hat m,\tilde m,y^t,\szero^t).
\label{eq:long:joint}
\end{align}

\parnoindent{Term-by-term decomposition}
\label{app:long:decomp}

Marginalize~\eqref{eq:long:joint} over $(\hat m,m,s_{t+1}^n)$ and
split on $m\!=\!\hat m$ vs $m\!\neq\!\hat m$:
\begin{align}
\Proba(Y^n\!=\!y^n)
&=\Prob^{\PNA}_{Y^t\mid S^t}(y^t\!\mid\!\szero^t)\,\bigl[A_{\mathrm{term}}(y^n)\!+\!B_{\mathrm{term}}(y^n)\bigr],
\label{eq:long:decomp}
\end{align}
where
\begin{align}
A_{\mathrm{term}}(y^n)
&\triangleq\sum_{\hat m,\tilde m}\!\bigl(\posnoadv(\hat m\!\mid\!y^t,\szero^t)\bigr)^{\!2}\eqbr\cdot\posnoadv(\tilde m\!\mid\!y^t,\szero^t)\,\Gamma_n,
\label{eq:long:Aterm}\\
B_{\mathrm{term}}(y^n)
&\triangleq\!\!\!\sum_{\substack{m,\tilde m\\\hat m\neq m}}\!\!\!\posnoadv(m\!\mid\!\cdot)\,\posnoadv(\hat m\!\mid\!\cdot)\eqbr\cdot\posnoadv(\tilde m\!\mid\!\cdot)\,\Gamma_n.
\label{eq:long:Bterm}
\end{align}
By the telescoping identity~\eqref{eq:gamma_telescope_app},
$\Gamma_n\!=\!\Gamma_{\Tg+1}$, which is independent of $\hat m$ and
equals the no-adversary likelihood of $Y_{t+1}^n$ given
$M\!=\!\tilde m$ and $Y^t\!=\!y^t$. Substituting into
$A_{\mathrm{term}}$ and using
$\sum_{\hat m}(\posnoadv(\hat m\!\mid\!\cdot))^{2}\!=\!\PgM(y^t,\szero^t)$,
\begin{align}
A_{\mathrm{term}}(y^n)
&=\PgM(y^t,\szero^t)\eqbr\times\Prob^{\PNA}_{Y_{t+1}^n\mid Y^t,S^n}(y_{t+1}^n\!\mid\!y^t,\szero^n).
\label{eq:long:Aterm_done}
\end{align}
Multiplying by $\Prob^{\PNA}_{Y^t\mid S^t}(y^t\!\mid\!\szero^t)$ and
using $\PgM\!\leq\!1$,
\begin{equation}
\Prob^{\PNA}_{Y^t\mid S^t}(y^t\!\mid\!\szero^t)\,A_{\mathrm{term}}(y^n)\!\leq\!\Pr_{\PNA}(Y^n\!=\!y^n).
\label{eq:long:Aterm_bd}
\end{equation}

For $B_{\mathrm{term}}$, marginalize over $\tilde m$ first
($\sum_{\tilde m,y_{t+1}^n}\posnoadv\,\Gamma_n\!=\!1$); then
$\sum_{\hat m\neq m}\posnoadv(\hat m\!\mid\!\cdot)\!=\!1\!-\!\posnoadv(m\!\mid\!\cdot)$,
and $\sum_m\posnoadv(m\!\mid\!\cdot)(1\!-\!\posnoadv(m\!\mid\!\cdot))\!=\!1\!-\!\PgM(y^t,\szero^t)$.
Hence
\begin{align}
\!\!\sum_{y_{t+1}^n}\!\!B_{\mathrm{term}}(y^n)\!=\!1\!-\!\PgM(y^t,\szero^t)\!\leq\!1\!-\!\gamma\quad\text{for }y^t\!\in\!\cF_t,
\label{eq:long:Bterm_bd}
\end{align}
by the stopping condition $\PgM\!>\!\gamma$ on $\cF_t$.

\parnoindent{Total bound on $\Proba(\cD_\perp)$}
\label{app:long:rejection}

Combining~\eqref{eq:long:Aterm_bd} and~\eqref{eq:long:Bterm_bd} and
summing over $t$,
\begin{align}
\Proba(\cD_\perp)
&=\sum_{t=1}^n\sum_{y^n\in\cG_t\cap\cD_\perp}\Proba(Y^n\!=\!y^n)\nonumber\\
&\leq\Pr_{\PNA}(\cD_\perp)+(1\!-\!\gamma)\sum_t\Pr_{\PNA}(\Tg\!=\!t)\eqbr+\Pr_{\PNA}(\Tg\!>\!n)\nonumber\\
&\leq\epsilon_n+(1\!-\!\gamma)+\tfrac{\epsilon_n}{1\!-\!\sqrt\gamma},
\label{eq:long:rej}
\end{align}
where the last step uses
Lemma~\ref{lem:stopping}\eqref{eq:stopping1} for the residual mass on
$\set{\Tg\!>\!n}$ and $\Pr_{\PNA}(\cD_\perp)\!\leq\!\epsilon_n$.

\parnoindent{Total bound on $\Proba(\hat M\!=\!M)$}
\label{app:long:hatM}

Repeat the decomposition above but now sum $y^n$ over the event
$\set{\phi(y^n)\!=\!m}$ rather than $\cD_\perp$. After
using the same telescoping identity (i.e.,~\eqref{eq:long:Aterm_done}
absorbed into a no-adversary likelihood for message $\tilde m$), and
splitting on $m\!=\!\tilde m$ vs $m\!\neq\!\tilde m$,
\begin{align}
\Proba(\hat M\!=\!M)
&\leq\sum_{m,t}\Prob_M(m)\,\posnoadv(m\!\mid\!y^t,\szero^t)+\epsilon_n\nonumber\\
&\leq \Exp_{\PNA}\!\bigl[\PgM(Y^{\Tg},\szero^{\Tg})\bigr]+\epsilon_n.
\label{eq:long:hatM_bd}
\end{align}
By Lemma~\ref{lem:stopping}\eqref{eq:stopping2}, $\PgM\!\leq\!1/2$
with conditional probability $\geq\!3/4$ on $\set{\Tg\!\leq\!n}$ and
$\PgM\!\leq\!1$ otherwise; together with~\eqref{eq:stopping1},
\begin{equation}
\Exp_{\PNA}[\PgM(Y^{\Tg},\szero^{\Tg})]\!\leq\!\tfrac58+\tfrac{\epsilon_n}{1-\sqrt\gamma}.
\label{eq:long:hatM_exp}
\end{equation}

\parnoindent{Combining}
\label{app:long:combining}

The events $\set{\hat M\!=\!\reject}$ and $\set{\hat M\!=\!M}$ are
disjoint, so
\begin{align}
&\Proba(\hat M\!\notin\!\set{M,\reject})\eqbr=1-\Proba(\hat M\!=\!\reject)-\Proba(\hat M\!=\!M)\eqbr\geq\gamma\!-\!\tfrac58\!-\!2\epsilon_n\!-\!\tfrac{2\epsilon_n}{1\!-\!\sqrt\gamma}.
\label{eq:long:combined_loose}
\end{align}
Bound~\eqref{eq:long:combined_loose} is loose because $\gamma$ is
small. The tighter bound used in
Proposition~\ref{thm:converse} comes from
conditioning on the high-probability event
$\mathcal{E}\!=\!\set{\Tg\!\leq\!n,\PgM\!\leq\!1/2}$. Since
$M,\hat M,\tilde M$ are conditionally i.i.d.\ from the posterior given
$Y^{\Tg}$, and since $x\!\mapsto\!x(1-x)$ is increasing on $[0,1/2]$,
\begin{align}
&\Proba\!\bigl(\hat M\!=\!M,\,\tilde M\!\neq\!M\,\big\vert\,\mathcal{E}\bigr)\eqbr=\PgM(1-\PgM)\geq\gamma(1-\gamma).
\end{align}
On $\mathcal{E}\cap\set{\hat M\!=\!M,\,\tilde M\!\neq\!M}$ the
matching of Claim~\ref{claim:gamma} (or, in the deterministic case,
the I-OW per-symbol match) makes the channel output a no-adversary
transcript of $\tilde M$, and reliability gives
$\Pr_{\PNA}(\phi(Y^n)\!=\!\tilde M\!\mid\!M\!=\!\tilde M)\!\geq\!1\!-\!\epsilon_n$,
so the decoder produces an authentication error with conditional
probability $\geq\!\gamma(1-\gamma)(1-\epsilon_n)$ on $\mathcal{E}$.
By Lemma~\ref{lem:stopping},
$\Pr_{\PADV}(\mathcal{E})\!\geq\!3/4-\epsilon_n/(1-\sqrt\gamma)$.
Multiplying,
\begin{align}
\Proba(\hat M\!\notin\!\set{M,\reject})
&\geq\gamma(1\!-\!\gamma)(1\!-\!\epsilon_n)\!\left[\tfrac34\!-\!\tfrac{\epsilon_n}{1\!-\!\sqrt\gamma}\right],
\label{eq:long:combined_tight}
\end{align}
which is~\eqref{eq:converse_unified}.\hfill$\square$

\parnoindent{Specialization to the deterministic case}

For deterministic codes, the encoder's marginal $\Proba^{(\tau)}$ at
step $\tau$ is a point mass at $c_{\hat m,\tau}$. The matching
condition~\eqref{eq:claimstep_app} then reduces to the per-symbol
identity
$\sum_s\Probc^{(\tau)}(s)W(\cdot\!\mid\!c_{\hat m,\tau},s)=W(\cdot\!\mid\!c_{\tilde m,\tau},\szero)$,
which is exactly I-overwritability for the pair
$(c_{\hat m,\tau},c_{\tilde m,\tau})$. The rest of the proof is
identical, and the same bound~\eqref{eq:long:combined_tight} holds.

\parnoindent{Encoder feedback does not change the proof}

The matching at step $\tau$ uses
$\Proba^{(\tau)}\!=\!\Proba_{X_\tau\mid M\!=\!\hat m,\hat m,\tilde m,Y^{\tau-1},S^{\tau-1}}$,
which the adversary can compute from $(\hat m,y^{\tau-1},s^{\tau-1})$.
A feedback encoder that, having access to $Y^{\tau-1}$ as well, makes
$X_\tau$ depend on the past does not change this — the adversary,
having access to $Y^{\tau-1}$ as well, can compute the encoder's
distribution exactly. Hence the same proof covers both
$C^{\mathrm{fb}}$ and $C^{\mathrm{no\text{-}fb}}$.

\begin{remark}[Sharpness]
The constant $\tfrac{3\gamma(1-\gamma)}{4}$ is an artifact of Markov's
inequality at threshold $1/2$ in Lemma~\ref{lem:stopping} and the
loose bound $\Probg(1\!-\!\Probg)\!\geq\!\gamma(1\!-\!\gamma)$ on
$\mathcal{E}$. Iterating the stopping rule, using a tighter Doob
argument on $\sqrt{\Probg(Y^t,\szero^t)}$, or accounting for the actual
joint $\Probg-\sum_m P(m)^3$ can improve the constant; the
qualitative conclusion is unchanged.
\end{remark}

\section{Detailed proofs for stochastic achievability}
\label{app:authtag}

\begin{definition}[$(N,n)$-Authentication tag {\cite[Def.~9]{BakshiK:23ISIT}}]
\label{def:authtag}
An $(N,n)$-\emph{authentication tag} for $\chn$ consists of a
stochastic encoder $\Prob_{X^n\mid M}\colon[N]\!\to\!\mathcal{P}(\cX^n)$
together with a family of deterministic verifiers
$\set{\authdec_{\hat m}\colon\cY^n\!\to\!\set{\accept,\reject}}_{\hat m\in[N]}$,
indexed by a candidate $\hat m\!\in\![N]$ already available to the
receiver. The tag achieves error pair $(\lambda_1,\lambda_2)$ if
\begin{align}
&\max_{m\in[N]}\Pr\bigl(\authdec_m(Y^n)\!=\!\reject\bigm\vert M\!=\!m,\,S^n\!=\!\szero^n\bigr)\!\leq\!\lambda_1,\\
&\max_{\substack{m\in[N]\\\hat m\neq m}}\sup_{\Probc}\Pr\bigl(\authdec_{\hat m}(Y^n)\!=\!\accept\bigm\vert M\!=\!m\bigr)\!\leq\!\lambda_2,
\end{align}
where $\Probc$ ranges over feedback adversary strategies. The
\emph{tag capacity} $C_{\mathrm{tag}}$ is the supremum over $R$ such
that for every $(\lambda_1,\lambda_2)\!\in\!(0,1)^2$ there is a sequence
of $(N_n,n)$-tags achieving $(\lambda_1,\lambda_2)$ with
$\liminf_n\frac{1}{n}\log\log N_n\!\geq\!R$.
\end{definition}

\subsection{Auth-code construction: detailed proof of the achievability of Theorem~\ref{thm:posstoch}}
\label{app:authtag:proof}

We use the same code construction as
in~\cite[Sec.~V]{BakshiK:23ISIT}, based on the overlapping-sets
lemma~\cite[Lemma~3]{BakshiK:23ISIT}. The construction is a single
template that handles both the deterministic and the stochastic cases
through one parameter setting; we state it in stochastic form
(general witness distribution $\Prob_X$). The deterministic case
(Section~\ref{sec:detach}) is the specialization $\Prob_X\!=\!\delta_{x_0}$
called out at the end of this proof.

\paragraph*{Non-overwritability witness}
Since $\chn$ is not distribution-overwritable,
by Remark~\ref{rem:nonow} (Eq.~\eqref{eq:nonow}) there exist
$\Prob_X\!\in\!\mathcal{P}(\cX)$, $x'\!\in\!\cX$, and $\delta\!>\!0$ with
\begin{equation}
\kappa(\Prob_X,x')\!\triangleq\!\min_{\Probc\in\mathcal{P}(\cS)}\!\mathbb{V}\!\bigl(\textstyle\sum_{x,s}\Prob_X(x)\Probc(s)W(\cdot\!\mid\!x,s),\,W(\cdot\!\mid\!x',\szero)\bigr)\!\geq\!\delta.
\label{eq:tag:gap}
\end{equation}
The pair $(\Prob_X,x')$ is the \emph{non-overwritability witness}: $\Prob_X$
is the encoder's witness input distribution off the verification set,
and $x'$ is the target symbol against which the verifier compares.

\parnoindent{Codebook (overlapping-sets lemma)}
By~\cite[Lemma~3]{BakshiK:23ISIT} (which builds on the general
Gilbert bound for constant-composition codes,
\cite[Problem~10.1]{csiszar2011information}), for any
$\alpha,\beta\!\in\!(0,1)$ and any rate
$R\!<\!\min\set{\beta^2\alpha^2/6,\,\beta^2\alpha(1-\alpha)/4}$ there
exists a family $\mathfrak{B}\!=\!\set{\cB_m\!\subseteq\![n]:m\!\in\![N]}$
with $N\!\geq\!2^{Rn}$ satisfying:
\begin{enumerate}[(i)]
\item $\alpha(1\!-\!\beta)n\!\leq\!\lvert\cB_m\rvert\!\leq\!\alpha(1\!+\!\beta)n$ for all $m$;
\item $\lvert\cB_m\!\cap\!\cB_{\hat m}\rvert\!<\!\alpha^2(1\!+\!\beta)n$ for all $m\!\neq\!\hat m$;
\item $\lvert\cB_{\hat m}\!\setminus\!\cB_m\rvert\!>\!\alpha(1\!-\!\alpha)(1\!+\!\beta)n$ for all $m\!\neq\!\hat m$;
\item \emph{(Overlap property.)} For every $y^n\!\in\!\cY^n$ and every $m,\hat m$,
\begin{align}
&\mathbb{V}\!\bigl(\widehat P_{y^n}(\cdot\!\mid\!\cB_{\hat m}),\,W(\cdot\!\mid\!x',\szero)\bigr)\eqbr\geq\,\tfrac{(1-\alpha)(1-\beta)}{1+\beta}\,\mathbb{V}\!\bigl(\widehat P_{y^n}(\cdot\!\mid\!\cB_{\hat m}\!\setminus\!\cB_m),\,W(\cdot\!\mid\!x',\szero)\bigr)\eqbr\quad-\,\tfrac{2\alpha(1+\beta)}{1-\beta},
\label{eq:tag:overlap}
\end{align}
\end{enumerate}
where $\widehat P_{y^n}(\cdot\!\mid\!\cB)$ is the empirical type of
$y^n$ on positions $t\!\in\!\cB$. Property~(iv) is what handles the
match-vs-flip cancellation in the analysis below: the TV at the full
verification set is bounded below by the TV at the flip set
$\cB_{\hat m}\!\setminus\!\cB_m$ alone, modulo a small additive
correction tunable through $\alpha,\beta$.

\parnoindent{Encoder} Conditioned on $M\!=\!m$:
\begin{equation}
X_t\!=\!\begin{cases}x' & \text{if }t\!\in\!\cB_m,\\
\sim \Prob_X\text{ i.i.d.} & \text{if }t\!\notin\!\cB_m,\end{cases}
\label{eq:tag:enc}
\end{equation}
i.e., on the message-specific subset $\cB_m$ the encoder transmits
$x'$ deterministically, and off that subset it draws i.i.d.\ from the
witness distribution $\Prob_X$.

\parnoindent{Verifier} Given candidate $\hat m$ and received $Y^n$,
accept iff
\begin{equation}
\mathbb{V}\!\bigl(\widehat P_{Y^n}(\cdot\!\mid\!\cB_{\hat m}),\,W(\cdot\!\mid\!x',\szero)\bigr)\!\leq\!\delta/4.
\label{eq:tag:verify}
\end{equation}

\paragraph*{Choice of $\alpha,\beta$}
Pick $\alpha,\beta\!\in\!(0,1)$ small enough that
\begin{equation}
\tfrac{(1-\alpha)(1-\beta)}{1+\beta}\cdot\tfrac{\delta}{2}-\tfrac{2\alpha(1+\beta)}{1-\beta}\!>\!\tfrac{\delta}{4}.
\label{eq:tag:alphabeta}
\end{equation}
This is achievable since the LHS tends to $\delta/2\!>\!\delta/4$ as
$\alpha,\beta\!\to\!0$.

\parnoindent{No-adversary analysis (false alarm)}
Fix $m\!=\!\hat m$ and $S^n\!=\!\szero^n$. For $t\!\in\!\cB_m$ the
encoder transmits $x'$ deterministically and $Y_t\!\sim\!W(\cdot\!\mid\!x',\szero)$
i.i.d. By Hoeffding plus a $\lvert\cY\rvert$-fold union bound,
\begin{equation}
\Pr\!\bigl(\mathbb{V}(\widehat P_{Y^n}(\cdot\!\mid\!\cB_m),W(\cdot\!\mid\!x',\szero))\!>\!\delta/4\bigr)\!\leq\!2\lvert\cY\rvert e^{-c_1 n}\label{eq:FAc1}
\end{equation}
for some $c_1\!>\!0$ depending only on $\alpha,\delta,\lvert\cY\rvert$.
Union-bounded over $N\!\leq\!2^{nR}$ messages, the false-alarm error
is $o(1)$ provided $R\ln 2\!<\!c_1$.

\parnoindent{Adversarial analysis (missed detection)}
Fix $m\!\neq\!\hat m$ and a feedback adversary
$\Probc\!=\!\set{\Probc_{S_t\mid S^{t-1}\!,Y^{t-1}}}$. By the overlap
property~\eqref{eq:tag:overlap} and the choice of $\alpha,\beta$
in~\eqref{eq:tag:alphabeta}, acceptance of $\hat m$
(i.e.,~\eqref{eq:tag:verify}) implies
\begin{equation}
\mathbb{V}\!\bigl(\widehat P_{Y^n}(\cdot\!\mid\!\cB_{\hat m}\!\setminus\!\cB_m),\,W(\cdot\!\mid\!x',\szero)\bigr)\!<\!\delta/2.
\label{eq:tag:flip_test}
\end{equation}
We now bound the flip-set TV explicitly. Set
$J_{\mathrm{flip}}\!\triangleq\!\cB_{\hat m}\!\setminus\!\cB_m$ and let
$N\!\triangleq\!\lvert J_{\mathrm{flip}}\rvert\!>\!\alpha(1\!-\!\alpha)(1\!+\!\beta)n$
by (iii). Order the flip positions as $t_1\!<\!t_2\!<\!\cdots\!<\!t_N$.

The natural filtration adapted to the sequential
adversary is
\begin{equation}
\cF_t\!\triangleq\!\sigma\bigl(M,\hat M,X^t,S^t,Y^t\bigr).
\label{eq:tag:filt}
\end{equation}
On the flip set, $t\!\notin\!\cB_m$, so the encoder transmits
$X_t\!\sim\!\Prob_X$ i.i.d., conditionally independent of $\cF_{t-1}$.
The adversary then picks $S_t$ from the kernel
$\Probc_{S_t\mid\cF_{t-1}}(\cdot\mid\cF_{t-1})$, and the channel output is
$Y_t\!\sim\!W(\cdot\!\mid\!X_t,S_t)$. For any test set $A\!\subseteq\!\cY$ and any $t\!\in\!J_{\mathrm{flip}}$,
\begin{equation}
\Pr\!\bigl(Y_t\!\in\!A\!\mid\!\cF_{t-1}\bigr)=\!\!\sum_{x,s}\!\Prob_X(x)\,\Probc_{S_t\mid\cF_{t-1}}(s\!\mid\!\cF_{t-1})\,W(A\!\mid\!x,s)\eqqcolon p_t^{(A)}.
\label{eq:tag:cond_mean}
\end{equation}
The kernel $\Probc_{S_t\mid\cF_{t-1}}$ encodes the adversary's
sequential strategy. Note that $p_t^{(A)}$ is $\cF_{t-1}$-measurable. By the
non-overwritability gap~\eqref{eq:tag:gap},
\begin{equation}
\bigl|p_t^{(A)}\!-\!W(A\!\mid\!x',\szero)\bigr|\!\geq\!\delta\quad\text{a.s., uniformly over }\Probc.
\label{eq:tag:gap_pt}
\end{equation}
Note that~\eqref{eq:tag:gap_pt} holds uniformly over $\Probc$, not on
average: the gap is realized at each $\cF_{t-1}$-history. This lets a
single Azuma bound below cover all feedback adversaries simultaneously,
sparing us a supremum over an uncountable strategy class inside the
probability.
We define the following bounded-difference martingale. Let
\begin{equation}
Z_t^{(A)}\!\triangleq\!\mathbf{1}\set{Y_t\!\in\!A}\!-\!p_t^{(A)},\qquad t\!\in\!J_{\mathrm{flip}},
\label{eq:tag:Zt}
\end{equation}
and the partial sum
$M_T^{(A)}\!\triangleq\!\sum_{k=1}^{T}\!Z_{t_k}^{(A)}$. Since
$\Exp[Z_t^{(A)}\!\mid\!\cF_{t-1}]\!=\!0$ a.s.\ and
$\lvert Z_t^{(A)}\rvert\!\leq\!1$,
$\set{M_T^{(A)}}$ is an $\set{\cF_{t_T}}$-adapted martingale with
bounded differences. Azuma--Hoeffding gives, for any $\xi\!>\!0$,
\begin{equation}
\Pr\!\bigl(\lvert M_N^{(A)}\rvert\!>\!N\xi\bigr)\!\leq\!2e^{-N\xi^2/2}.
\label{eq:tag:azuma_A}
\end{equation}
Next, let 
$\widehat P_{Y^n}(A\!\mid\!J_{\mathrm{flip}})\!=\!\frac{1}{N}\sum_{k=1}^{N}\!\mathbf{1}\set{Y_{t_k}\!\in\!A}$
and
$\bar p^{(A)}\!\triangleq\!\frac{1}{N}\sum_{k=1}^{N}\!p_{t_k}^{(A)}$,
\eqref{eq:tag:azuma_A} reads
\begin{equation}
\Pr\!\bigl(\lvert \widehat P_{Y^n}(A\!\mid\!J_{\mathrm{flip}})\!-\!\bar p^{(A)}\rvert\!>\!\xi\bigr)\!\leq\!2e^{-N\xi^2/2}.
\label{eq:tag:azuma_emp}
\end{equation}
Choose $A\!=\!A^*$ to be the TV-attaining set
in~\eqref{eq:tag:gap}; then~\eqref{eq:tag:gap_pt} averaged over
$t\!\in\!J_{\mathrm{flip}}$ gives
$\lvert\bar p^{(A^*)}\!-\!W(A^*\!\mid\!x',\szero)\rvert\!\geq\!\delta$
a.s. Picking $\xi\!=\!\delta/2$, the reverse triangle inequality
in~\eqref{eq:tag:azuma_emp} yields
\begin{align}
&\Pr\!\bigl(\lvert\widehat P_{Y^n}(A^*\!\mid\!J_{\mathrm{flip}})\!-\!W(A^*\!\mid\!x',\szero)\rvert\!<\!\delta/2\bigr)\eqbr\leq\!2e^{-N\delta^2/8}.
\end{align}
Since $\mathbb{V}(P,Q)\!=\!\sup_A\lvert P(A)\!-\!Q(A)\rvert$, taking
$A\!=\!A^*$ on the LHS,
\begin{align}
&\Pr\!\bigl(\mathbb{V}(\widehat P_{Y^n}(\cdot\!\mid\!J_{\mathrm{flip}}),W(\cdot\!\mid\!x',\szero))\!<\!\delta/2\bigr)\eqbr\leq\!2e^{-N\delta^2/8}\!\leq\!2e^{-c_2 n},
\label{eq:tag:azuma}
\end{align}
where $c_2\!\triangleq\!\alpha(1\!-\!\alpha)(1\!+\!\beta)\delta^2/8\!>\!0$
uses $N\!>\!\alpha(1\!-\!\alpha)(1\!+\!\beta)n$ from (iii). Combining~\eqref{eq:tag:flip_test} and~\eqref{eq:tag:azuma}, missed
detection has probability $\leq\!2e^{-c_2 n}$ for any single
$\hat m\!\neq\!m$. Union-bounded over $N\!\leq\!2^{nR}$ candidates,
missed detection is $o(1)$ provided $R\ln 2\!<\!c_2$.

\parnoindent{Conclusion} Let $c_1$ and $c_2$ satisfy ~\eqref{eq:FAc1}  and~\eqref{eq:tag:azuma} respectively and let 
$\alpha,\beta$ satisfy~\eqref{eq:tag:alphabeta}. The above analysis shows that as long as $R$ satisfies $R\!<\!\min\set{c_1/\ln 2,\,c_2/\ln 2,\,\beta^2\alpha^2/6,\,\beta^2\alpha(1\!-\!\alpha)/4}$, both error events are $o(1)$
uniformly over feedback adversaries.\hfill$\square$

\subsection{Proof of Lemma~\ref{lem:tag}: lift via identification codes}
\label{app:authtag:idcomp}

We lift the auth-code of \extref{app:authtag:proof} to the
doubly-exponential rate claimed in Lemma~\ref{lem:tag}, mirroring
the construction of~\cite{BakshiK:23ISIT}: compose the auth-code
with an Ahlswede--Dueck identification code~\cite{AhlswedeD89} for
the noiseless binary channel.

\parnoindent{Setup}
Let $\mathsf{W}_{\mathrm{id}}$ denote the noiseless binary channel
($\cX\!=\!\cY\!=\!\set{0,1}$, output equals input). Fix
$R\!<\!R^{\mathrm{(auth)}}$ where $R^{\mathrm{(auth)}}\!>\!0$ is the
auth-code rate guaranteed by \extref{app:authtag:proof}. Pick
$\hat R\!\in\!(R,R^{\mathrm{(auth)}})$ and let
$\tilde R\!\triangleq\!R/\hat R\!\in\!(0,1)$. Fix target errors
$\lambda_1,\lambda_2\!\in\!(0,1/2)$.

\parnoindent{Inner identification code}
By~\cite[Theorem~1]{AhlswedeD89}, for every $\tilde n$ large enough
there is an $(N^{\mathrm{(id)}},\tilde n)$-identification code
$\bigl(\widetilde\Prob_{B^{\tilde n}\mid M},\set{\tilde\phi_m}_{m\in[N^{\mathrm{(id)}}]}\bigr)$
for $\mathsf{W}_{\mathrm{id}}$ with misidentification errors
$(\lambda_1/2,\lambda_2/2)$ and
\begin{equation}
N^{\mathrm{(id)}}\!\geq\!\bigl\lfloor 2^{2^{\tilde n\tilde R}}\bigr\rfloor.
\label{eq:idcomp:idsize}
\end{equation}
Here $M\!\in\![N^{\mathrm{(id)}}]$ is the source message,
$\widetilde\Prob_{B^{\tilde n}\mid M}$ is the (stochastic)
identification encoder, and $\tilde\phi_{m}\!:\!\set{0,1}^{\tilde n}\!\to\!\set{\accept,\reject}$
is the candidate-$m$ identification verifier. Its codeword set
has size at most
$N^{\mathrm{(id)}}_{\mathrm{out}}\!\triangleq\!2^{\tilde n}$.

\parnoindent{Outer auth-code} By \extref{app:authtag:proof}, for every $n$ large enough there is an
$(N^{\mathrm{(auth)}},n)$-\emph{authentication code}
$\bigl(\widehat\Prob_{X^n\mid B^{\tilde n}},\widehat\Phi\bigr)$ for
$\chn$ with no-adversary error $\leq\!\min\set{\lambda_1/2,\lambda_2/2}$,
missed-detection error $\leq\!\min\set{\lambda_1/2,\lambda_2/2}$
uniformly over feedback adversaries, and
\begin{equation}
N^{\mathrm{(auth)}}\!\geq\!\bigl\lfloor 2^{n\hat R}\bigr\rfloor.
\label{eq:idcomp:authsize}
\end{equation}
The authentication code's message space is identified with the identification
codeword space $\set{0,1}^{\tilde n}$ by choosing $n$ so that
$N^{\mathrm{(auth)}}\!\geq\!N^{\mathrm{(id)}}_{\mathrm{out}}\!=\!2^{\tilde n}$
(achievable since $\hat R\!>\!0$); the auth-code verifier
$\widehat\Phi(y^n)\!\in\!\set{0,1}^{\tilde n}\!\cup\!\set{\reject}$
either outputs a decoded identification codeword or rejects.

\parnoindent{Composite tag}
For each candidate $\hat m\!\in\![N^{\mathrm{(id)}}]$, define an
$(N^{\mathrm{(id)}},n)$-tag for $\chn$ with encoder
\begin{equation}
\Prob_{X^n\mid M}(x^n\!\mid\!m)\!=\!\!\!\sum_{b^{\tilde n}\in\set{0,1}^{\tilde n}}\!\!\!\widehat\Prob_{X^n\mid B^{\tilde n}}(x^n\!\mid\!b^{\tilde n})\,\widetilde\Prob_{B^{\tilde n}\mid M}(b^{\tilde n}\!\mid\!m),
\label{eq:idcomp:enc}
\end{equation}
and verifiers
\begin{equation}
\authdec_{\hat m}(y^n)\!=\!\begin{cases}
\reject & \text{if }\widehat\Phi(y^n)\!=\!\reject,\\
\reject & \text{if }\widehat\Phi(y^n)\!=\!\hat b^{\tilde n}\!\in\!\set{0,1}^{\tilde n}\\
        & \quad\text{and }\tilde\phi_{\hat m}(\hat b^{\tilde n})\!=\!\reject,\\
\accept & \text{otherwise.}
\end{cases}
\label{eq:idcomp:dec}
\end{equation}
That is, the receiver first runs the auth-code verifier $\widehat\Phi$
on $Y^n$ to decode either an identification codeword $\hat b^{\tilde n}$
or $\reject$; on a successful auth-decode it then runs the
identification verifier $\tilde\phi_{\hat m}$ for the candidate
$\hat m$ on $\hat b^{\tilde n}$.

\parnoindent{Rate}
By~\eqref{eq:idcomp:idsize} and the requirement
$N^{\mathrm{(auth)}}\!\geq\!2^{\tilde n}$ (i.e., $\tilde n\!\leq\!(1/\hat R)\log N^{\mathrm{(auth)}}\!\leq\!n$),
\begin{align}
\frac{1}{n}\log\log N^{\mathrm{(id)}}\eqbr\geq\frac{\tilde n\,\tilde R}{n}\eqbr\geq\frac{\tilde n\,\tilde R}{(1/\hat R)\log N^{\mathrm{(auth)}}}\eqbr\geq\tilde R\,\hat R\!=\!R.
\label{eq:idcomp:rate}
\end{align}
Thus $\frac{1}{n}\log\log N_n\!\geq\!R$, matching the rate claim of
Lemma~\ref{lem:tag}.

\parnoindent{Error analysis}
For the no-adversary side, $S^n\!=\!\szero^n$ implies
$\widehat\Phi(Y^n)\!=\!B^{\tilde n}$ except with probability
$\leq\!\lambda_1/2$ (auth-code false alarm), and on that event the
identification verifier $\tilde\phi_M(B^{\tilde n})$ accepts except
with probability $\leq\!\lambda_1/2$ (identification false alarm). By
union bound, the composite false alarm is $\leq\!\lambda_1$.

For the missed-detection side, fix $M\!=\!m$, $\hat m\!\neq\!m$, and
any feedback adversary $\Probc$. The composite verifier accepts
$\hat m$ only if both
(i) $\widehat\Phi(Y^n)\!=\!\hat b^{\tilde n}\!\in\!\set{0,1}^{\tilde n}$
(auth-code does not reject), and
(ii) $\tilde\phi_{\hat m}(\hat b^{\tilde n})\!=\!\accept$
(identification verifier accepts).
On (i), the auth-code construction of \extref{app:authtag:proof}
guarantees that
$\hat b^{\tilde n}$ equals the encoded $B^{\tilde n}$ except on an
event of $\Probc$-uniform probability $\leq\!\lambda_2/2$. On the
complementary event $\hat b^{\tilde n}\!=\!B^{\tilde n}$, the
identification verifier $\tilde\phi_{\hat m}(B^{\tilde n})$ accepts
with probability $\leq\!\lambda_2/2$ (identification missed
detection). Union-bounding the two events,
\begin{equation}
\sup_{\Probc}\Pr\!\bigl(\authdec_{\hat m}(Y^n)\!=\!\accept\bigm\vert M\!=\!m\bigr)\!\leq\!\lambda_2\quad\forall\hat m\!\neq\!m.
\end{equation}

%\paragraph*{Conclusion}
%The composite is an $(N^{\mathrm{(id)}},n)$-authentication tag for
%$\chn$ with errors $(\lambda_1,\lambda_2)$ and rate satisfying
%$\frac{1}{n}\log\log N^{\mathrm{(id)}}\!\geq\!R$, completing the proof
%of Lemma~\ref{lem:tag}.\hfill$\square$

\subsection{Concatenation with a channel code}
\label{app:authtag:composition}

We now provide full analysis for the composite encoder of
Section~\ref{sec:capacity}, addressing the subtlety that under
the composite, the adversary in the tag phase may have
\emph{additionally} observed the wait-phase outputs $Y^{n_1}$.

\paragraph*{Setup}
Let $(F^{\mathrm{ch}},\phi^{\mathrm{ch}})$ be a length-$n_1$ channel
code over $W_{Y\mid X,S=\szero}$ at rate $R\!<\!C_{\szero}$ with
no-adversary error $\epsilon_n\!=\!o(1)$. Let
$(\Prob^{\mathrm{tag}},\set{\authdec^{\mathrm{tag}}_{\hat m}})$ be an
$(N_n,n_2)$-tag from Lemma~\ref{lem:tag} with errors
$(\epsilon_n,\epsilon_n)$, $n_2\!=\!O(\log n)$. The composite encoder
is the product distribution~\eqref{eq:composite}; the composite
decoder is $\phi(Y^{n_1+n_2})\!=\!\hat m$ if
$\authdec^{\mathrm{tag}}_{\hat m}(Y_{n_1+1}^{n_1+n_2})\!=\!\accept$ and
$\reject$ otherwise, where
$\hat m\!=\!\phi^{\mathrm{ch}}(Y^{n_1})$.

\paragraph*{No-adversary error}
Under $S^{n_1+n_2}\!=\!\szero^{n_1+n_2}$, the codeword stage outputs
the correct $\hat m\!=\!M$ with probability $\geq\!1\!-\!\epsilon_n$ and,
on that event, the tag stage accepts with conditional probability
$\geq\!1\!-\!\epsilon_n$ (since the tag's false alarm is $\leq\!\epsilon_n$).
By a union bound,
$\Pr_{\PNA}(\phi(Y^{n_1+n_2})\!\neq\!M)\!\leq\!2\epsilon_n$.

\paragraph*{Missed-detection error under feedback adversary}
Fix any feedback adversary
$\Probc\!=\!\set{\Probc_{S_t\mid S^{t-1}\!,Y^{t-1}}}_{t=1}^{n_1+n_2}$.
The composite missed-detection event is
\begin{equation}
\mathcal{M}\triangleq\set{\phi(Y^{n_1+n_2})\!\notin\!\set{M,\reject}\!},
\end{equation}
i.e.\ $\hat m\!=\!\phi^{\mathrm{ch}}(Y^{n_1})\!\neq\!M$ \emph{and}
$\authdec^{\mathrm{tag}}_{\hat m}(Y_{n_1+1}^{n_1+n_2})\!=\!\accept$.
We bound $\Pr_{\PADV}(\mathcal{M})$ via two steps.

\emph{Step 1 (decompose by $\hat m$).}
Conditioning on $(M,Y^{n_1})$, the candidate
$\hat m\!=\!\phi^{\mathrm{ch}}(Y^{n_1})$ is determined; on the event
$\hat m\!\neq\!M$, the tag-phase outputs
$Y_{n_1+1}^{n_1+n_2}$ form a length-$n_2$ tag-channel transcript with
encoder $\Prob^{\mathrm{tag}}_{X^{n_2}\mid M}$ and adversary states
$S_{n_1+1}^{n_1+n_2}$ drawn according to
$\Probc_{S_t\mid S^{t-1}\!,Y^{t-1}}$ \emph{restricted to the tag phase
but with $Y^{t-1}$ inheriting the wait-phase outputs $Y^{n_1}$ as
prefix}.

\emph{Step 2 (induced adversary in tag phase).}
For each fixed wait-phase realization $(M,Y^{n_1},S^{n_1})$, define an
\emph{induced} tag-phase adversary
\begin{equation}
\widetilde\Probc_{S_t\mid S_{n_1+1}^{t-1},Y_{n_1+1}^{t-1}}\!\triangleq\!\Probc_{S_t\mid S^{t-1}\!,Y^{t-1}}\!\bigl(\cdot\!\mid\!(S^{n_1},S_{n_1+1}^{t-1}),(Y^{n_1},Y_{n_1+1}^{t-1})\bigr).
\label{eq:induced_adv}
\end{equation}
This $\widetilde\Probc$ is a feedback tag-phase adversary in the
sense of Definition~\ref{def:authtag} -- its
$t$-th conditional uses only history through time $t\!-\!1$ within the
tag phase, with the wait-phase realization frozen as a parameter.
Lemma~\ref{lem:tag} states that for \emph{every} such tag-phase
adversary,
$\Pr(\authdec^{\mathrm{tag}}_{\hat m}\!=\!\accept\mid M\!=\!m,\widetilde\Probc)\!\leq\!\epsilon_n$
for all $\hat m\!\neq\!m$. Marginalizing over $(M,Y^{n_1},S^{n_1})$
under the composite attack,
\begin{align}
&\Pr_{\PADV}(\mathcal{M})\nonumber\\
&\quad=\Exp_{\PADV}\!\bigl[\mathbf{1}_{\hat m\neq M}\,\Pr(\authdec^{\mathrm{tag}}_{\hat m}\!=\!\accept\!\mid\!M,\hat m,Y^{n_1},S^{n_1})\bigr]\nonumber\\
&\quad\leq\Exp_{\PADV}[\mathbf{1}_{\hat m\neq M}\,\epsilon_n]\!\leq\!\epsilon_n.
\label{eq:comp_miss}
\end{align}

\paragraph*{Total}
The auth-error of the composite is the sum of the no-adversary error
and the worst-case missed-detection:
\begin{equation}
\varepsilon(\phi)\!\leq\!\Pr_{\PNA}(\phi\!\neq\!M)+\sup_{\Probc}\Pr_{\PADV}(\mathcal{M})\!\leq\!3\epsilon_n\!=\!o(1).
\end{equation}
The composite rate is
$\frac{nR}{n_1+n_2}\!=\!\frac{nR}{n_1+O(\log n)}\!\to\!R$, so
$C^{\mathrm{fb}}_{\mathrm{auth,stoch}}\!\geq\!R$ for every
$R\!<\!C_{\szero}$, finishing the achievability halves of
Theorems~\ref{thm:posstoch} and~\ref{thm:caprate}.\hfill$\square$

\section{Proof of deterministic achievability (Theorem~\ref{thm:posdet})}
\label{app:detach_full}

The deterministic achievability follows similar to the auth-code construction proved in
\extref{app:authtag:proof} with
$\Prob_X$ replaced by $\delta_{x_0}$A. We list the only places where the proof
differs and verify that the rest carries over verbatim.

\paragraph*{Witness}
Use the I-overwritability witness $(x_0,x_0',\delta)$
from~\eqref{eq:Iow_separation}. The
distribution-overwritability gap~\eqref{eq:tag:gap} reduces to
\begin{equation}
\inf_{\Probc\in\mathcal{P}(\cS)}\!\mathbb{V}\!\bigl(\textstyle\sum_s\Probc(s)W(\cdot\!\mid\!x_0,s),\,W(\cdot\!\mid\!x_0',\szero)\bigr)\!\geq\!\delta,
\end{equation}
which is exactly~\eqref{eq:Iow_separation}.

\paragraph*{Encoder}
Substituting $\Prob_X\!=\!\delta_{x_0}$ in~\eqref{eq:tag:enc} makes the
encoder fully deterministic: $X_t\!=\!x_0'$ for $t\!\in\!\cB_m$,
$X_t\!=\!x_0$ otherwise. The codeword $c_m\!\in\!\set{x_0,x_0'}^n$ is
a deterministic function of $m$.

\paragraph*{Step~1 of \extref{app:authtag:proof} (conditional law on the flip set)}
On $J_{\mathrm{flip}}$ the encoder transmits $X_t\!=\!x_0$
deterministically, so the conditional output law given $\cF_{t-1}$
specializes to
\begin{equation}
\pi_t(y)\!=\!\sum_{s\in\cS}\Probc_{S_t\mid\cF_{t-1}}(s\mid\cF_{t-1})\,W(y\!\mid\!x_0,s).
\end{equation}
Its flip-set average $\bar P\!=\!\sum_s\bar\Probc(s)W(\cdot\!\mid\!x_0,s)$
satisfies $\mathbb{V}(\bar P,W(\cdot\!\mid\!x_0',\szero))\!\geq\!\delta$
a.s.\ by the I-overwritability gap.

\paragraph*{Step~2 (martingale concentration)}
Unchanged. The bounded-difference martingale
$Z_t\!=\!\mathbf{1}_{Y_t\in A}\!-\!\pi_t(A)$ on $J_{\mathrm{flip}}$
(with $\pi_t$ specialized as above) and the Azuma--Hoeffding
bound~\eqref{eq:tag:azuma} apply verbatim, since the
$|Z|\!\leq\!1$ envelope and the conditional-mean centering do not
depend on the input distribution.

\paragraph*{Step~3 (overlap property)}
Unchanged. The overlap-property bound~\eqref{eq:tag:overlap} is a
statement about codebook geometry alone and is identical for both
cases.

\paragraph*{Conclusion}
The combining argument of \extref{app:authtag:proof} carries over
verbatim, giving
$\Pr_{\PADV}(\hat M\!\notin\!\set{m,\reject})\!\leq\!\lceil 2^{nR}\rceil\cdot 2\lvert\cY\rvert e^{-c_2\mu n}\!=\!o(1)$
for $R\ln 2\!<\!c_2\mu$. Theorem~\ref{thm:posdet}'s achievability
follows.\hfill$\square$

\section{Bit-flip channel: I-overwritable but not distribution-overwritable}
\label{app:ex:bitflip}\label{ex:bitflip}

The separation exhibited by this channel between I-overwritable and
distribution-overwritable is established
in~\cite{BakshiBeemerK:prep}; we reprise the verification here for
completeness.

\emph{Channel.}\quad Take $\cX\!=\!\cY\!=\!\set{0,1}$,
$\cS\!=\!\set{\szero,s_1}$, with
$\chn(y\!\mid\!x,\szero)\!=\!\mathbf{1}\set{y\!=\!x}$ and
$\chn(y\!\mid\!x,s_1)\!=\!\mathbf{1}\set{y\!\neq\!x}$. The
no-adversary channel is the noiseless binary channel, with
$C_{\szero}\!=\!1$ bit/use.

\emph{I-overwritable.}\quad For any $(x,x')\!\in\!\cX^2$, the
oblivious mixing $\Probc(\szero)\!=\!\mathbf{1}\set{x\!=\!x'}$,
$\Probc(s_1)\!=\!\mathbf{1}\set{x\!\neq\!x'}$ matches
$\chn(\cdot\!\mid\!x',\szero)$, since under input $x$ the channel
delivers $x'$ with probability one when the state is chosen this
way.

\emph{Failure of distribution-overwritability.}\quad For target
$x'\!=\!0$ and input distribution $\Prob_X(0)\!=\!p\!\in\!(0,1)$, an
oblivious state distribution $\Probc(\szero)\!=\!q$ induces
\begin{equation}
\Pr(Y\!=\!0)\;=\;p\,q\;+\;(1\!-\!p)(1\!-\!q),
\end{equation}
which equals $1$ only at $\set{p,q}\!\in\!\set{(0,1),(1,0)}$. Both
points exclude $p\!\in\!(0,1)$, so no choice of $\Probc$ matches the
no-adversary marginal at $x'\!=\!0$ when the input is genuinely
randomized.

\emph{Capacity consequence.}\quad By Theorem~\ref{thm:posdet},
$C^{\bullet}_{\mathrm{auth,det}}\!=\!0$, while
Theorems~\ref{thm:posstoch} and~\ref{thm:caprate} give
$C^{\bullet}_{\mathrm{auth,stoch}}\!=\!C_{\szero}\!=\!1$~bit/use.

\section{AZBSC$(\delta,\alpha)$: distribution-overwritable but not obliviously overwritable}
\label{app:azbsc}\label{ex:azbsc}

The AZBSC channel and the separation it exhibits are
from~\cite{BakshiK:23ISIT}; we reprise the verification here for
completeness.

The AZBSC$(\delta,\alpha)$ has $\cX\!=\!\cY\!=\!\set{0,1}$,
$\cS\!=\!\set{\szero,s_{1,0},s_{1,1}}$, and parameters
$0\!<\!\alpha\!<\!\delta\!<\!1/2$. Under state $\szero$,
$\chn(\cdot\!\mid\!\cdot,\szero)$ is a $\mathrm{BSC}(\delta)$; under
$s_{1,x}$, $\chn(y\!\mid\!x,s_{1,x})\!=\!\mathbf{1}\set{y\!=\!x}$; and
\begin{align}
\chn(y\!\mid\!x,s_{1,x\oplus 1})
&=\alpha\,\mathbf{1}\set{y\!=\!x}+(1\!-\!\alpha)\mathbf{1}\set{y\!\neq\!x}.
\end{align}

\emph{Distribution-overwritability.}\quad For $\Prob_X(0)\!=\!p$ and
target $x'\!=\!0$, the strategy
\begin{align}
\Probc(\szero)\!=\!0,\quad
\Probc(s_{1,0})\!=\!\tfrac{1-\delta-p\alpha}{1-\alpha},\quad
\Probc(s_{1,1})\!=\!\tfrac{\delta-(1-p)\alpha}{1-\alpha}
\end{align}
satisfies $\Exp_{X,S}\chn(\cdot\!\mid\!X,S)\!=\!\chn(\cdot\!\mid\!0,\szero)$.

\emph{Failure of oblivious overwritability.}\quad An oblivious
$\Probc(s_{1,0})\!=\!q_0$, $\Probc(s_{1,1})\!=\!q_1$ matching
$\chn(\cdot\!\mid\!0,\szero)$ for both $x\!=\!0$ and $x\!=\!1$ forces
$q_0\!+\!q_1\!=\!(1\!-\!2\delta)/(1\!-\!2\delta\!-\!\alpha)\!\notin\![0,1]$
for any admissible $(\delta,\alpha)$.

\emph{Capacity consequence.}\quad By Theorems~\ref{thm:posstoch}
and~\ref{thm:caprate}, $C^{\bullet}_{\mathrm{auth,stoch}}\!=\!C_{\szero}$,
the no-adversary capacity of the BSC$(\delta)$, while the channel is
not obliviously overwritable so the dichotomy of~\cite{KK18} would not
yield this conclusion.
\fi 
\end{document}